\documentclass[aap,MSNbibl,citesort,dvips]{arximspdf}
\usepackage{graphicx}
\doi{10.1214/10-AAP722}
\volume{21}
\issue{3}
\pubyear{2011}
\firstpage{1136}
\lastpage{1179}

\makeatletter
\newcommand{\eqref}[1]{(\ref{#1})}
\renewcommand{\phi}{\varphi}

\newcommand{\N}{\mathbb{N}}
\newcommand{\R}{\mathbb{R}}
\newcommand{\C}{\mathbb{C}}
\newcommand{\E}{\mathbb{E}}
\renewcommand{\P}{\mathbb{P}}
\newcommand{\U}{\mathcal{U}}
\newcommand{\D}{\mathcal{D}}
\newcommand{\M}{\mathcal{M}}
\newcommand{\F}{\mathcal{F}}
\newcommand{\mcC}{\mathcal{C}}
\renewcommand{\S}{\mathcal{S}}
\newcommand{\V}{\mathcal{V}}
\newcommand{\ol}{\overline}
\newcommand{\gammat}{\tilde{\gamma}}
\newcommand{\fhat}{\hat{f}}
\newcommand{\isom}{\simeq}
\newcommand{\trace}{\operatorname{Tr}}
\newcommand{\tracenorm}{\operatorname{tr}}
\newcommand{\Wg}{\operatorname{Wg}}
\newcommand{\Mob}{\operatorname{Mob}}
\newcommand{\Rem}{\operatorname{Rem}}
\newcommand{\I}{\mathrm{I}}
\newcommand{\id}{\operatorname{id}}
\renewcommand{\Re}{\operatorname{Re}}
\renewcommand{\Im}{\operatorname{Im}}
\newcommand{\iy}{\infty}

\newtheorem{theorem}{Theorem}[section]
\newproclaim{definition}[theorem]{Definition}
\newtheorem{proposition}[theorem]{Proposition}
\newproclaim{remark}[theorem]{Remark}
\newtheorem{lemma}[theorem]{Lemma}
\newtheorem{conjecture}[theorem]{Conjecture}
\newproclaim{question}[theorem]{Question}
\newtheorem{corollary}[theorem]{Corollary}
\makeatother

\begin{document}
\begin{frontmatter}

\title{Gaussianization and eigenvalue statistics for random quantum
channels (III)}
\runtitle{Random quantum channels III}

\begin{aug}
\author[A]{\fnms{Beno\^{i}t} \snm{Collins}\corref{}\ead[label=e1]{bcollins@uottawa.ca}\thanksref{t1,t2}}
\and
\author[B]{\fnms{Ion} \snm{Nechita}\ead[label=e2]{inechita@uottawa.ca}\thanksref{t2}}
\pdfauthor{Benoit Collins, Ion Nechita}
\runauthor{B. Collins and I. Nechita}
\affiliation{University of Ottawa}
\thankstext{t1}{Supported in part by ANR GranMa and ANR Galoisint.}
\thankstext{t2}{Supported in part by NSERC
Grant RGPIN/341303-2007.}
\address[A]{D\'epartement de Math\'ematique\\
\quad et Statistique\\
Universit\'e d'Ottawa\\
585 King Edward\\
Ottawa, Ontario K1N6N5\\
Canada\\
and\\
CNRS\\
Institut Camille Jordan\\
Universit\'e Lyon 1\\
43 Bd du 11 Novembre 1918\\
69622 Villeurbanne\\
France\\
\printead{e1}} %adresu isvedimo komanda gale!
\address[B]{D\'epartement de Math\'ematique\\
\quad et Statistique\\
Universit\'e d'Ottawa\\
585 King Edward\\
Ottawa, Ontario K1N6N5\\
Canada\\
\printead{e2}}
\end{aug}

% HISTORY:
\received{\smonth{4} \syear{2010}}

% ABSTRACT
%
\begin{abstract}
In this paper, we present applications of the calculus developed in
Collins and Nechita
[\textit{Comm. Math. Phys.} \textbf{297} (2010) 345--370] and obtain
an exact
formula for the moments of random quantum channels whose
input is a pure state thanks to Gaussianization methods.
Our main application is an in-depth study of the random matrix model
introduced by Hayden and Winter [\textit{Comm. Math. Phys.} \textbf
{284} (2008) 263--280]
and used recently by Brandao and Horodecki [\textit{Open Syst. Inf.
Dyn.} \textbf{17} ({2010}) 31--52] and
Fukuda and King [\textit{J. Math. Phys.} \textbf{51} ({2010}) 042201]
to refine the Hastings counterexample to the
additivity conjecture in quantum information theory. This model is
exotic from the point of
view of random matrix theory as its eigenvalues obey two different
scalings simultaneously.
We study its asymptotic behavior and obtain an asymptotic expansion for
its von Neumann entropy.
\end{abstract}

% KEYWORDS
%

\end{frontmatter}

%s1 ###
\section{Introduction}

In the paper~\cite{collins-nechita-1} we developed a calculus
permitting the computation of any moments
of random quantum channels.
It has already proven useful in understanding the random matrix models involved
in the additivity violation theorems and in
improving lower bounds of dimensions needed to obtain
violation of the additivity of entropy estimates (developed in \cite
{collins-nechita-1,collins-nechita-2,collins-nechita-4}), as well as
in the study of random quantum states associated with graphs \cite
{collins-nechita-zyczkowski}.
In the present work we study two more applications of our calculus, to
new random matrix models
introduced for quantum information theoretic purposes.

The first application is of theoretical interest and of a nonasymptotic
nature: we extend our calculus
to Gaussian matrices and show that it yields explicit formulas for the
moments of Wishart matrices and of outputs\vadjust{\goodbreak} of random quantum
channels. The formulas are of a purely combinatorial nature and make it
possible to bypass Weingarten calculus, whose
asymptotic estimates can be involved.
For this we use a ``Gaussianization'' method.

The second application is an extended
study of the random matrix model that was introduced by Hayden and
Winter in~\cite{hayden-winter}
and used recently in   \mbox{\cite
{fukuda-king,fukuda-king-moser,brandao-horodecki}} to refine the results
of Hastings~\cite{hastings}.
As a motivation, let us recall the quantum information theoretic
context of this random matrix. A~\textit{quantum channel} is a linear completely positive
trace-preserving map $\Phi\dvtx \M_n(\C)\to\M_k(\C)$.
A~\textit{density matrix} is a self-adjoint positive semidefinite
matrix with trace $1$.
Let $\Delta_k = \{x \in\R_+^k \mid\sum_{i=1}^k x_i = 1\}$ be the
$(k-1)$-dimensional probability simplex.
The \textit{Shannon entropy} of $x$ is defined to~be
\[
H(x) = -\sum_{i=1}^k x_i \log x_i.
\]
These definitions are extended to density matrices by functional calculus:
\[
H(\rho) = -\trace\rho\log\rho.
\]
For a quantum channel $\Phi\dvtx \M_n(\C) \to\M_k(\C)$, its \textit
{minimum output entropy} is defined by
\[
H_{\min}(\Phi) = \mathop{\mathop{\min}_{\rho\in\M_n(\C) }}_{
\rho\geq
0, \trace{\rho}=1} H(\Phi(\rho)).
\]
The additivity conjecture for minimum output entropies is arguably one
of the most important in quantum information theory, and it
can be stated as follows.
\begin{conjecture}
For all quantum channels $\Phi_1$ and $\Phi_2$, we have
%
%
%e1 ###
\begin{equation}\label{eq:additivity_H1}
H_{\min}(\Phi_1 \otimes\Phi_2) = H_{\min}(\Phi_1) + H_{\min
}(\Phi_2).
\end{equation}
\end{conjecture}

This conjecture was disproven by Hastings in~\cite{hastings} as follows.
\begin{theorem}
There exists a counterexample to the conjecture for the choice $\Phi
_1=\overline{\Phi_2}$.
\end{theorem}

In the proof of~\cite{hastings}, one reason why $\Phi_1=\overline
{\Phi_2}$ yields a counterexample is that it ensures that the largest
eigenvalue
of outputs of well-chosen inputs---Bell states---is much bigger than
the other eigenvalues.
The counterexamples to the additivity conjecture obtained thus far use
a random matrix model which we redefine in Section~\ref{subsec:conjugate}
and call $Z_n$.
The main result of this paper is as follows (the dimension ratio $c$ is
a fixed positive constant).

\begin{theorem}
As $n \to\iy$, $k \sim cn$, the eigenvalues $\lambda_1 \geq\cdots
\geq\lambda_{n^2}$ of $Z_n$ are such that:
\begin{itemize}
\item
in probability, $cn\lambda_1\to1$;\vadjust{\goodbreak}
\item
almost surely,
$\frac{1}{n^2-1}\sum_{i=2}^{n^2}\delta_{c^2n^2\lambda_i}$ converges
to a Marchenko--Pastur distribution of parameter $c^{2}$;
\item
almost surely,
\[
H(Z_n) =
\cases{\displaystyle
2\log n - \frac{1}{2c^2}+o(1), & \quad   if   $c\geq1$,\cr\displaystyle
2\log(cn) - \frac{c^2}{2}+o(1), & \quad   if   $0<c<1$,
}
\]
as $n\to\infty$, where $H$ is the von Neumann entropy.
\end{itemize}
\end{theorem}

The interest of this result is that
it yields improvements to the results of
\cite{fukuda-king,fukuda-king-moser,brandao-horodecki,hastings}, as
the only data that
these papers were using was a lower bound on the largest eigenvalue of $Z_n$,
whereas the above theorem gives a full understanding of the eigenvalue
behavior of
$Z_n$.

In addition, the matrix model $Z_n$ has the
novel property that it has two different regimes for its eigenvalues
(one in $n^{-1}$ and one in $n^{-2}$). As far as we know,
it is the first model in random matrix theory whose eigenvalues have
two regimes simultaneously.

The proof of the main theorem
uses a mix of moment methods and functional calculus methods.
It is very instructive, as the moment method is used
to prove the convergence in distribution of the eigenvalues of smaller
decay, and this goes beyond the standard
intuition that moment methods instead give results about the larger eigenvalues.
Actually, our Theorem~\ref{thm:asmpt-QZQ} shows
new kinds of cancellation properties,
going beyond those which are usually expected with standard
``moments--cumulants'' and
``connectedness'' arguments.

This paper is organized as follows.
We first recall some known facts about Wick calculus, Weingarten
calculus and noncommutative and
free probability theory. We also recall our graphical calculus
introduced in~\cite{collins-nechita-1}
and extend it to Gaussian graphical calculus.
We use this to obtain new nonasymptotic results for the moments of some
single random channels.
We obtain further asymptotic results in the single random channel
setting, and we
then return to the random matrix model introduced in the bi-channel
setting by
Hayden and Winter, computing the asymptotics of the subleading eigenvalues.

%s2 ###
\section{Wick calculus and Weingarten calculus}

In this section we recall known results which allow the computation of
expectations
against Gaussian measures and Haar measures on unitary groups, as well
as some standard facts in free probability theory.

%s2.1 ###
\subsection{Wick calculus}

A \textit{Gaussian space} $V$ is a real vector space of random variables
with moments of all orders such that each of these random variables
are centered Gaussian distributions.
Such a Gaussian space comes with a positive symmetric bilinear form
$(x,y)\to\E[xy]$.\vadjust{\goodbreak}
Gaussian spaces are in one-to-one cor\-respondence
with Euclidean spaces, and isomorphisms of Gaussian spaces
correspond to the notion of isomorphisms of Euclidean spaces.
In particular, the Euclidean norm of a random variable
determines it fully (via its variance)  and if two
\mbox{random} variables are given, then their joint distribution is
determined by their angle.
The following is usually called the Wick lemma.

\begin{lemma}\label{lem:wick-lemma}
Let $V$ be a Gaussian space and $x_{1}, \ldots, x_k$ be elements in~$V$.
If $k=2l+1,$ then $\E[x_1\cdots x_k]=0,$ and
if $k=2l,$ then
%
%
%e2 ###
\begin{equation}\label{eq:Wick}
\E[x_1\cdots x_k]=\mathop{\mathop{\sum}_{p=\{\{i_1,j_1\},\ldots
, \{
i_l,j_l\}\} }}_{
\mathrm{pairing\ of\ }\{1,\ldots,k\} } \prod
_{m=1}^l \E[x_{i_m}x_{j_m}].
\end{equation}
In particular, it follows that if $x_1,\ldots, x_p$ are independent
standard Gaussian random variables, then
\[
\E[x_1^{k_1}\cdots x_{p}^{k_p}]=\prod_{i=1}^p (2k_i)!! .
\]
\end{lemma}

For a proof see, for instance,~\cite{zvonkin}.
It is possible to extend the notion of a Gaussian space to a complex
Gaussian space. A~complex-valued vector space $V$ is called a Gaussian space if and
only if
for any real structure on $V$,
the pair $(\Re(V), \Im(V))$ is a real-valued Gaussian space.
It can be readily checked that in the case of a complex Gaussian space,
the Wick Lemma~\ref{lem:wick-lemma} holds with exactly the same statement.

We will usually denote by $G_{n,m}$ (or $G$ when there is no ambiguity)
the standard complex Gaussian random matrix $n\times m$.
It has the distribution
$\exp(-N\times  \trace(GG^*))\,dG,$ where $dG$ is the Lebesgue measure on the space of the
$n\times m$ complex matrices properly rescaled, and $G^*=\overline
{G}^t$ is
the standard algebraic adjoint operator.

Since we shall mostly be concerned with traces of products of random
matrices in this paper, we need to introduce one last item of notation
for generalized traces, which we borrow from~\cite{bryc}. For some
matrices $A_1, A_2, \ldots, A_s \in\M_n(\C)$, some permutation
$\sigma\in\S_p$ and some function $t\dvtx \{1, \ldots, p\} \to\{1,
\ldots, s\},$ we define
\[
\trace_{\sigma, t}(A_1, \ldots, A_s) = \prod_{c\in\mcC(\sigma
)}\trace\biggl(\operatorname{\vec{\prod}}\limits_{j\in c}A_{t(j)} \biggr),
\]
where $\mcC(\sigma)$ is the set of cycles of $\sigma$. When $s=p$,
we use the simplified notation $\trace_{\sigma, t}(A_1, \ldots, A_p)
= \trace_{\sigma, \id}(A_1, \ldots, A_p)$. We also put $\trace
_\sigma(A) = \trace_\sigma(A, A,\break \ldots, A)$.

%s2.2 ###
\subsection{Weingarten calculus}

In this section, we recall a few facts about Weingarten calculus.\vadjust{\goodbreak}

\begin{definition}
The unitary Weingarten function
$\Wg(n,\sigma)\dvtx \N\times\break \bigcup_{p \in\N^*} \S_p \to\R$
is a function of a dimension parameter $n$ and of a permutation $\sigma$
in the symmetric group $\S_p$.
It is the pseudo-inverse of the function $\sigma\mapsto n^{\# \sigma
}$ under the convolution for the symmetric group ($\# \sigma$ denotes
the number of cycles of the permutation $\sigma$).
\end{definition}

Notice that the function $\sigma\mapsto n^{\# \sigma}$ is invertible
as $n \geq p$.
We refer to~\cite{collins-sniady} for historical references and
further details.
We shall use the shorthand notation $\Wg(\sigma) = \Wg(n, \sigma)$
when the dimension parameter $n$ is obvious.

The function $\Wg$ is used to compute integrals with respect to
the Haar measure on the unitary group.

\begin{theorem}
\label{thm:Wg}
Let $n$ be a positive integer and
$(i_1,\ldots,i_p)$, $(i'_1,\ldots,i'_p)$,
$(j_1,\ldots,j_p)$, $(j'_1,\ldots,j'_p)$
be $p$-tuples of positive integers from $\{1, 2, \ldots, n\}$. Then,
%
%e3 ###
\begin{eqnarray}
\label{bid} &&\int_{\U(n)}U_{i_1j_1} \cdots U_{i_pj_p}
\overline{U_{i'_1j'_1}} \cdots
\overline{U_{i'_pj'_p}}\,dU\nonumber
\\[-8pt]
\\[-8pt]
&& \qquad =
\sum_{\sigma, \tau\in\S_{p}}\delta_{i_1i'_{\sigma(1)}}\cdots
\delta_{i_p i'_{\sigma(p)}}\delta_{j_1j'_{\tau(1)}}\cdots
\delta_{j_p j'_{\tau(p)}} \Wg(n,\tau\sigma^{-1}).
\nonumber
\end{eqnarray}

If $p\neq p',$ then
%
%
%e4 ###
\begin{equation} \label{eq:Wg_diff} \int_{\U(n)}U_{i_{1}j_{1}}
\cdots
U_{i_{p}j_{p}} \overline{U_{i'_{1}j'_{1}}} \cdots
\overline{U_{i'_{p'}j'_{p'}}}\,dU= 0.
\end{equation}
\end{theorem}

We are interested in the values of the Weingarten function in the limit
$n \to\iy$. The following result encloses all the data we need for
our computations
relating to the asymptotics of the $\Wg$ function; see \cite
{collins-imrn} for a proof.

\begin{theorem}\label{thm:mob} For a permutation $\sigma\in\S_p$,
let $\operatorname{Cycles}(\sigma)$ denote the set of cycles of $\sigma$. Then,
%
%
%e5 ###
\begin{equation}
\Wg(n,\sigma)= \prod_{c\in\operatorname{Cycles} (\sigma)}\Wg(n,c)\bigl(1+O(n^{-2})\bigr)
\end{equation}
and
%
%
%e6 ###
\begin{equation}
\Wg(n,(1,\ldots,d) ) = (-1)^{d-1}c_{d-1}\prod_{-d+1\leq j \leq
d-1}(n-j)^{-1},
\end{equation}
where $c_i=\frac{(2i)!}{(i+1)! i!}$ is the $i$th Catalan number.
\end{theorem}

A shorthand for this theorem is the introduction of a function $\Mob$
on the symmetric
group, invariant under conjugation and multiplicative over the cycles,
satisfying,
for any permutation $\sigma\in\S_p$,
%
%
%e7 ###
\begin{equation}
\Wg(n,\sigma) = n^{-(p + |\sigma|)} \bigl(\Mob(\sigma) + O(n^{-2})\bigr),\vadjust{\goodbreak}
\end{equation}
where $|\sigma|=p-\# \sigma$ is the \textit{length} of $\sigma$,
that is, the minimal number of transpositions that multiply to $\sigma
$. We refer to~\cite{collins-sniady} for details about the function
$\Mob$.

%s2.3 ###
\subsection{Elementary review of noncommutative and free probability theory}
\label{reminders-wishart}

A~\textit{noncommutative probability space} is
an algebra $\mathcal A$ with unit endowed with a tracial
state $\phi$.
An element of $\mathcal A$ is called
a (noncommutative) random variable. In this paper we shall be mostly
concerned with the noncommutative probability space of \textit{random
matrices} $(\M_n(L^{\iy-}(\Omega, \P)), \E[n^{-1}\trace(\cdot
)])$ [we use the standard notation $L^{\iy-}(\Omega, \P) = \bigcap_{p\geq1} L^p(\Omega, \P)$].

Let $(a_1,\ldots,a_k)$ be a $k$ -tuple of self-adjoint random
variables and let
$\mathbb{C}\langle X_1 , \ldots,  X_k \rangle$ be the
free $*$-algebra of noncommutative polynomials on $\mathbb{C}$
generated by
the $k$ indeterminates $X_1, \ldots,X_k$.
The \textit{joint distribution} of the family $\{a_i\}_{i=1}^k$ is the
linear form
\begin{eqnarray*}
\mu_{(a_1,\ldots,a_k)} \dvtx  \C\langle X_1, \ldots,X_k \rangle
&\to&\C,
\\
P &\mapsto&\phi(P(a_1,\ldots,a_k)).
\end{eqnarray*}

Given a $k$-tuple $(a_1,\ldots,a_k)$ of free
random variables such that the distribution of $a_i$ is $\mu_{a_i}$,
the joint distribution
$\mu_{(a_1,\ldots,a_k)}$ is uniquely determined by the
$\mu_{a_i}$'s.
A~family $(a_1^{n},\ldots,a_k^{n})_n$ of $k$-tuples of random
variables is said to \textit{converge in distribution} toward
$(a_1,\ldots,a_k)$
iff for all $P\in\C\langle X_1, \ldots,X_k \rangle$,
$\mu_{(a_1^n,\ldots,a_k^n)}(P)$ converges toward
$\mu_{(a_1,\ldots,a_k)}(P)$ as $n\to\infty$.

The following result is from~\cite{nica-speicher} and will be crucial
for us. In what follows, $\mathit{NC}(p)$ denotes the set of noncrossing
partitions on $p$ elements, endowed with the reversed refinement
partial order (see~\cite{nica-speicher}, Lecture 9), which makes it
into a lattice.

\begin{lemma}\label{lem:S_p}
The function
$d(\sigma,\tau) = |\sigma^{-1} \tau|$ is an integer-valued distance
on $\S_p$. Further, it has the following properties:
\begin{itemize}
\item the diameter of $\S_p$ is $p-1$;
\item$d(\cdot, \cdot)$ is left and right translation invariant;
\item for three permutations $\sigma_1,\sigma_2, \tau\in\S_p$, the
quantity $d(\tau,\sigma_1)+d(\tau,\sigma_2)$
has the same parity as $d(\sigma_1,\sigma_2)$;
\item the set of geodesic points (elements which saturate the triangle
inequality) between the identity permutation $\id$ and some
permutation $\sigma\in\S_p$ is in bijection with the set of
noncrossing partitions smaller than $\pi$, where the partition $\pi$
encodes the cycle structure of $\sigma$. Moreover, the preceding
bijection preserves the lattice structure.
\end{itemize}
\end{lemma}

We finish by collecting the bare minimum of free probability theory
results needed for the development of the main results of this paper.
We skip the definition of freeness, as we will not need it.
Free cumulants are
multilinear\vadjust{\goodbreak} maps indexed by noncrossing partitions $\sigma\in \mathit{NC}(p)$
on $p$ elements
\[
\kappa_{\sigma}\dvtx   \underbrace{A\times\cdots\times A}_{p\ \mathrm{times}} \to\C
\]
such that
%
%
%e8 ###
\begin{equation}
\label{eq:moment-cumulant-formula}
\sum_{\pi\leq\sigma\in \mathit{NC}(p)}\kappa_{\pi}(x_1,\ldots,x_p)=\E
_{\sigma}[x_1,\ldots,x_p]
\end{equation}
for all noncrossing partitions $\sigma\in \mathit{NC}(p)$, where $\E_\sigma
[x_1,\ldots,x_p]$ is the product over the blocks
$\{x_{i_1},\ldots,x_{i_j}\}$ of $\sigma$, of
$\E[x_{i_1},\ldots, x_{i_j}]$.
Cumulants are known to be multiplicative over blocks and therefore
a special role is played by the cumulant corresponding to the maximal
partition $\mathbf{1}_p$, which we denote by
$\kappa(a_1,\ldots,a_p):=\kappa_{\mathbf{1}_p}(a_1,\ldots,a_p)$.

We will need free cumulants for computational purposes, in order to
identify free Poisson distributions.
Let us mention, for the benefit of the interested reader, that
the main property of the free cumulants is that mixed cumulants of free
variables vanish.

We recall that the \textit{free Poisson distribution} of parameter $c$
is given by
\[
\pi_c=\max(1-c,0)\delta_0+\frac{\sqrt{4c-(x-1-c)^2}}{2\pi x}
\mathbf{1}_{[1+c-2\sqrt{c},1+c+2\sqrt{c}]}(x) \,dx.
\]
It is characterized by the fact that all its free cumulants are equal
to $c$.
Although we will not need this fact, it is worth mentioning that it has
a semigroup structure with respect to the additive free convolution
of Voiculescu (see, e.g.,~\cite{nica-speicher}).
It is also sometimes called the Marchenko--Pastur distribution. One can
compute (minus) the entropy of this probability distribution:
%
%
%e9 ###
\begin{equation}\label{eq:entropy-free-poisson}
K_c = \int_{}^{} x \log x \,d\pi_c(x) =
\cases{\displaystyle
\frac{1}{2} + c \log c, & \quad    if   $c \geq1$,\cr\displaystyle
\frac{c^2}{2}, & \quad    if   $0<c<1$.
}
\end{equation}

%s3 ###
\section{Unitary and Gaussian graphical calculi}

In this section we briefly recall the results of \cite
{collins-nechita-1} for the convenience of the reader and in order to
make the paper self-contained. We then introduce the Gaussian graphical
calculus and present a first application of it to Wishart matrices.

%s3.1 ###
\subsection{Axioms of unitary graphical calculus}

The purpose of the graphical calculus introduced in \cite
{collins-nechita-1} is to yield an effective method to evaluate
the expectations of random tensors with respect to the Haar measure on
a unitary group.
The tensors under consideration can be constructed from a few
elementary tensors such as the Bell state,
fixed kets and bras, and random unitary matrices.
In graphical language, a~tensor corresponds to a \textit{box}, and an
appropriate Hilbertian structure yields a correspondence
between boxes and tensors.
However, the calculus yielding expectations only relies on diagrammatic
operations.

Each box $B$ is represented as a rectangle with decorations on its
boundary. The decorations are either white [elements of the set of
white decorations $S(B)$] or black [elements of the dual set of black
decorations, $S^*(B)$]. In the Hilbertian picture, decorations
correspond to complex vector spaces, dual decorations being associated
to dual spaces. Figure~\ref{fig:box} depicts an example of a box.

%
%f1 ###
\begin{figure}

\includegraphics{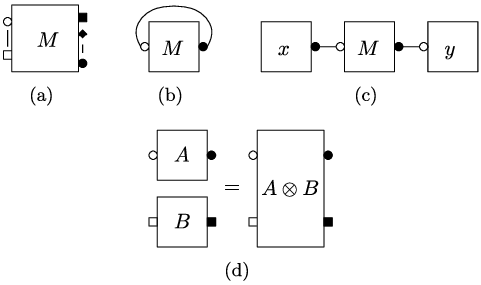}

\caption{Basic diagrams and axioms.}\label{fig:box}
\end{figure}

It is possible to construct new boxes from old ones by formal algebraic
operations such as sums or products.
We call a \textit{diagram} a picture consisting of boxes and wires
according to the following rule:
a wire may link a white decoration in $S(B)$ to its black counterpart
in $ S^*(B)$. A~diagram can be turned into a box by choosing an orientation and a
starting point.

Regarding the Hilbertian structure, wires correspond to tensor contractions.
There exists an involution for boxes and diagrams. It is antilinear and
turns a decoration
in $S(B)$ into its counterpart in $ S^*(B)$.
Our conventions are close to those of~\cite{coecke,jones}, and we hope
that they are familiar to
the reader acquainted with existing graphical calculi of various types
(planar algebra theory, Feynman diagram, traced category theory).
Our notation is designed to
conform well to the problem of computing expectations, as shown in the
next section. In Figure~\ref{fig:box}(b)--(d) we depict the trace of a matrix, multiplication
of tensors and the tensor product operation, respectively.
For details, we refer to~\cite{collins-nechita-1}.

%s3.2 ###
\subsection{Planar expansion}

In this subsection we describe the main application of our calculus.
For this, we need a concept of \textit{removal} of boxes $U$ and~$\ol U$. A~removal $r$ is a way to pair decorations of the $U$ and $\ol U$ boxes
appearing in a diagram.
It therefore consists of a pairing $\alpha$ of the white decorations
of $U$ boxes with the white decorations of $\ol U$ boxes,
together with a pairing $\beta$ between the black decorations of $U$
boxes and the black decorations of $\ol U$ boxes.
Assuming that $\D$ contains $p$ boxes of type $U$ and that the boxes
$U$ (resp., $\ol U$) are labeled from $1$ to $p$,
then $r=(\alpha,\beta),$ where $\alpha,\beta$ are permutations of
$\mathcal{S}_p$. The set of all removals of $U$ and $\ol U$ boxes is
denoted by $\Rem_U(\D)$.

Given a removal $r \in\Rem_U(\D)$ we construct a new diagram $\D_r$
associated with~$r$, one which has the important property that it no
longer contains boxes of type $U$ or $\ol U$.
We start by erasing the boxes $U$ and $\ol U$, but keep the decorations
attached to them.
Assuming that we have labeled the erased boxes $U$ and $\ol U$ with
integers from $\{1, \ldots, p\}$, we connect \textit{all} the (inner
parts of the) \textit{white} decorations of the $i$th erased $U$ box
with the corresponding (inner parts of the) \textit{white} decorations
of the $\alpha(i)$th erased $\ol U$ box. In a similar manner, we use
the permutation $\beta$ to connect black decorations.

In~\cite{collins-nechita-1}, we proved the following result.
\begin{theorem}\label{thm:Wg_diag}
\[
\E_U(\D)=\sum_{r=(\alpha, \beta) \in\Rem_U(\D)} \D_r \Wg(n,
\alpha\beta^{-1}).
\]
\end{theorem}

%s3.3 ###
\subsection{Gaussian planar expansion}\label{sec:gaussian}

We now consider the case where we allow a new special box $G$ in our
diagrams, corresponding to a \textit{Gaussian random matrix}. We shall
address the same issue as in the unitary case: computing the expected
value of a random diagram with respect to the Gaussian probability measure.

To begin, consider $\D,$ a diagram which contains, among other
constant tensors, boxes corresponding to independent Gaussian random
matrices of \textit{covariance one} (identity). We can deal with more
general Gaussian matrices by multiplying the standard ones by constant
matrices. Note that a box can appear several times, adjoints of boxes
are allowed and the diagram may be disconnected. Also, Gaussian
matrices need not be square.

The expectation value of such a random diagram $\D$ can be computed by
a \textit{removal} procedure, as in the unitary case. Without loss of
generality, we assume that we do not have adjoints of Gaussian matrices
in our diagram, but instead their complex conjugate box. This
assumption allows for a more straightforward use of the Wick Lemma \ref
{lem:wick-lemma}. As in the unitary case, we can assume that $\D$
contains only one type of random Gaussian box $G$; the other
independent random Gaussian matrices are assumed to be constant at this
stage as they shall afterward be removed in the same manner.

A removal of the diagram $\D$ is a pairing between \textit{Gaussian
boxes} $G$ and their conjugates $\ol G$. The set of removals is denoted
by $\Rem_G(\D),$ and it may be empty: if the number of $G$ boxes is
different from the number of $\ol G$ boxes, then $\Rem_G(\D) =
\varnothing$ [this is consistent with the first case of the Wick formula~\eqref{eq:Wick}]. Otherwise, a~removal $r$ can identified with a
permutation $\alpha\in\S_p$, where $p$ is the number of $G$ and $\ol
G$ boxes. Let us stress here the main difference between the notion of
a removal in the Gaussian and the Haar unitary cases. In the Haar
unitary (or the Weingarten) case, a~removal was associated with a
\textit{pair of permutations}: one had to pair white decorations of
$U$ and $\ol U$ boxes and, independently, black decorations of
conjugate boxes. On the other hand, in the Gaussian/Wick case, one
pairs {conjugate boxes}: white and black decorations are paired in an
identical manner, hence only one permutation is needed to encode the removal.

To each removal $r$ associated with a permutation $\alpha\in\S_p$
there corresponds a removed diagram $\D_r,$ constructed as follows. We
starts by erasing the boxes $G$ and $\ol G$, but keep the decorations
attached to these boxes. Then, the decorations (white \textit{and}
black) of the $i$th $G$ box are paired with the decorations of the
$\alpha(i)$th $\ol G$ box in a coherent manner; see Figure \ref
{fig:expectation_Gaussian}.

%
%f2 ###
\begin{figure}

\includegraphics{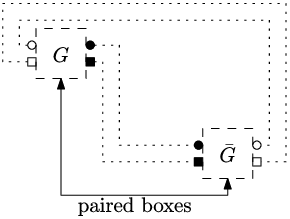}

\caption{Pairing of boxes in the Gaussian case.}
\label{fig:expectation_Gaussian}
\end{figure}

The graphical reformulation of the Wick Lemma~\ref{lem:wick-lemma}
becomes the following theorem, which we state without proof.

\begin{theorem}\label{thm:Wick_diag}
\[
\E_G[\D]=\sum_{r \in\Rem_G(\D)} \D_r .
\]
\end{theorem}

%s3.4 ###
\subsection{Moments of Wishart matrices}\label{sec:Wishart}
As a first application of our Gaussian graphical calculus, we compute
the moments of traces of products of Wishart matrices.
By definition, a~\textit{Wishart matrix} of parameters $(n,k)$ is a
positive random matrix $W \in\M_n(\C)$ such that
\[
W = G\cdot G^*,
\]
where $G \in\M_{n \times k}(\C)$ is a standard Gaussian random
matrix. In our graphical formalism, since we only consider
Gaussian random matrices, the previous equation corresponds to the
graphical substitution in Figure~\ref{fig:Wishart_to_Gaussian}; round
decorations correspond to $n$-dimensional complex Hilbert spaces $\C
^n$ and square-shaped labels correspond to $\C^k$.\vadjust{\goodbreak}

%
%f3 ###
\begin{figure}

\includegraphics{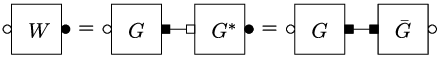}

\caption{Diagram of a Wishart matrix.}
\label{fig:Wishart_to_Gaussian}
\end{figure}

The same problem of computing expected values of traces of Wishart
matrices was considered in~\cite{bryc,graczyk,hanlon,mingo-nica},
and we shall rederive Corollary~3 of Theorem~2 from~\cite{bryc}. The
general covariance case (Theorem 2 in~\cite{bryc}) can be easily
derived from the result below.

\begin{proposition}\label{prop:moments_Wishart}
Let $W_1, W_2, \ldots, W_s$ be independent Wishart matrices with unit
covariance and parameters $(n,k_1)$, $(n,k_2), \ldots, (n,k_s),$
respectively. For a permutation $\sigma\in\S_p$ and a function $t:\{
1, \ldots, p\} \to\{1, \ldots, s\}$, we have
%
%
%e10 ###
\begin{equation}\label{eq:moments_Wishart}
\E[\trace_{\sigma,t}(W_1,\ldots,W_s)] = \sum_{\alpha\in\S
_p(t)}\prod_{j=1}^s k_j^{\#\alpha_j}n^{\#(\sigma^{-1}\alpha)},
\end{equation}
where $\S_p(t) = \{\alpha\in\S_n \mid t = t \circ\alpha\}$. Every
permutation $\alpha\in\S_p(t)$ leaves the level sets of $t$
invariant and induces on each set $t^{-1}(j)$ a permutation $\alpha_j$
($j = 1, \ldots, s$).
\end{proposition}

\begin{pf}
We consider the diagram $\D$ corresponding to the left-hand side of
equation~(\ref{eq:moments_Wishart}). It contains $n$ Wishart boxes
from the set $\{W_1, \ldots, W_s\}$ which are wired according to the
permutation $\sigma$ (see Figure~\ref{fig:prod_Wishart_bryc}).
Computing the expectation of the diagram $\D$ is rather
straightforward using our graphical calculus. Since we are dealing with
$s$ independent Gaussian matrices $G_1, \ldots, G_s$ (recall that $W_j
= G_j G_j^*$), we need to apply Theorem~\ref{thm:Wick_diag} $s$ times,
once for each Gaussian matrix $G_j$. Each box $G_j$ appears
$|t^{-1}(j)|$ times and, using Theorem~\ref{thm:Wick_diag}, we get
\[
\E[\D] = \sum\D_{\alpha_1, \ldots, \alpha_s},
\]
where each permutation $\alpha_j \in\S_{|t^{-1}(j)|}$ encodes the
removal procedure for the $G_j$ boxes.

%
%f4 ###
\begin{figure}[b]

\includegraphics{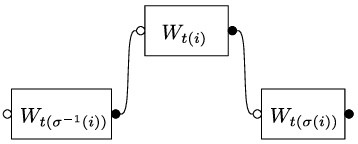}

\caption{Monomials of traces of Wishart matrices.}
\label{fig:prod_Wishart_bryc}
\end{figure}

Diagrams obtained after the successive removal procedures $\D_{\alpha
_1, \ldots, \alpha_s}$ are made of loops of two types: loops
associated with the $n$-dimensional space $\C^n$ and loops associated
with ``internal spaces'' $\C^{k_j}$. In order to count the number of
loops of each dimensionality, let us first observe that the set of
$s$-tuples of permutations $(\alpha_1, \ldots, \alpha_s)$ is in
bijection with the set of permutations $\alpha\in\S_p(t)$ defined in
the statement of the theorem.

For such a permutation $\alpha\in\S_p(t)$, let us count the number
of loops corresponding to traces over $\C^{k_j}$. Initially, the $p_j$
decorations of the $G_j$ boxes are connected in the simplest manner:
the $k_j$ decoration of the $i$th $G_j$ box is connected to the
corresponding decoration of the $\ol{G_j}$ box with the same index
$i$. The $j$th removal procedure, encoded by the permutation $\alpha
_j$, then produces a number of $\#(\id^{-1} \alpha_j) = \#\alpha_j$
loops. Hence, the contribution of the $\C^{k_j}$-type loops is $k_j^{\#
\alpha_j}$.

The computation of the loops associated with $\C^n$ is more involved
since the decorations are already nontrivially linked by the
permutation $\sigma$. Since $\sigma$ may not respect the level sets
of the function $t$, we need to consider the global action of~$\alpha
$, the restrictions $\alpha_j$ not being sufficient in this case.
Since the boxes are initially connected by $\sigma$ and the removal
procedures add wires according to the permutation~$\alpha$, the total
number of loops is $\#(\sigma^{-1} \alpha)$. Adding all loop
contributions, we obtain the announced formula~(\ref{eq:moments_Wishart}).
\end{pf}
\begin{remark}
We can consider more general covariances in the graphical model and
obtain Theorem 2 of~\cite{bryc} in its full generality. All there is
to be done is to add constant tensors associated with covariance
matrices in our diagrams. After the successive removal procedures, we
are left with loops \textit{and} traces of monomials in these constant
matrices. Since our purpose in this section was to illustrate the
Gaussian graphical calculus, we leave the details of this more
technical generalization to the interested reader.
\end{remark}

%s4 ###
\section{Application of Gaussianization: Pure states through random
quantum channels}

%s4.1 ###
\subsection{Single random channel model}

In this section we present an important application of the Gaussian
diagrammatic calculus: we compute eigenvalue statistics for the action
of a random quantum channel on a pure quantum state. By definition, a~\textit{quantum channel} $\Phi: \M_n(\C) \to\M_n(\C)$ is a
trace-preserving, completely positive map. According to the Stinespring
theorem, such a linear application can be written as
\[
\Phi(X) = \Phi^{U,Y}(X) = \trace_k [ U(X \otimes Y)U^* ],
\]
where $U$ is a unitary matrix in $\U(nk)$ and $Y$ is a $k$-dimensional
rank-one projector. A~diagrammatic representation of the above formula
is presented in Figure~\ref{fig:quantum_channel}.\vadjust{\goodbreak} The set of quantum
channels can be endowed with a natural probability measure by fixing
the projection $Y$ and picking $U$ uniformly with respect to the Haar
measure on the unitary group $\U(nk)$. This is the model of randomness
we refer to when we speak or random quantum channels, and it has
received a lot of attention from the quantum information community
\cite{collins-nechita-1,hayden-winter}. From the definition of $\Phi
$ we can see that the Weingarten calculus developed in \cite
{collins-nechita-1} may be applied to this situation since random
unitary matrices are a key element in the problem. However, when random
quantum channels are presented with rank-one inputs (or pure states),
we show that the simpler Gaussian calculus can be used, see Figure~\ref{fig:rank_one_Gaussian}. Using this
approach, we shall recover some exact formulas for the moments of the
output from~\cite{collins-nechita-1}, as well as some asymptotic
results from~\cite{nechita}.

%
%f5 ###
\begin{figure}

\includegraphics{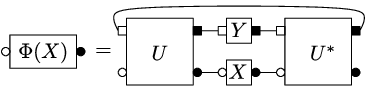}

\caption{Diagram for a quantum channel.}
\label{fig:quantum_channel}
\end{figure}

We are interested in the output random matrix
%
%
%e11 ###
\begin{equation}
\label{defZ}
Z = \Phi^{U,Y}(X),
\end{equation}
where $X$ is a rank-one projector.
The main result, obtained in~\cite{nechita}, is as follows.

\begin{proposition}
Let $W = G \cdot G^* \in\M_n(\C)$ be a Wishart matrix with
parameters $(n,k)$.
Then,
\[
Z = \Phi(X) = W / \trace(W).
\]
\end{proposition}

Observe that this result does not depend on the choice of $X,Y$ due to
the invariance of the Haar measure.

%
%f6 ###
\begin{figure}[b]

\includegraphics{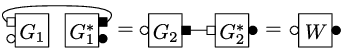}

\caption{An equivalent diagram for quantum channels with rank one $X$
and $Y.$}
\label{fig:rank_one_Gaussian}
\end{figure}

The main point is that we can show (see~\cite{nechita}) that the
eigenvalues of~$Z$, that is, the normalized eigenvalues of $W$, are
independent of the trace of $W$.
This implies that we we can apply the results on Wishart matrices
developed in Section~\ref{sec:Wishart} to this particular case.

%s4.2 ###
\subsection{Exact moments}

In this section we provide exact formulas for the moments $\E[\trace
(Z^p)]$ of the output of a random quantum channel. Other formulas for
the same quantities (as well as some recursion relations) have been
obtained in~\cite{nechita,sz1,zycsommers}.\vadjust{\goodbreak}

Using the Gaussianization trick, we have
\[
\E[\trace(Z^p)] = \frac{\E[\trace(W^p)]}{\E[\trace(W)^p]},
\]
where $W$ is a Wishart matrix with parameters $(n,k)$. One uses
Proposition~\ref{prop:moments_Wishart} to compute $\E[\trace(W^p)]$
and $\E[\trace(W)^p]$:
\begin{eqnarray*}
\E[\trace(W^p)] &=& \sum_{\alpha\in\S_p} k^{\# \alpha} n^{\#
(\gamma^{-1}\alpha)},\\
\E[\trace(W)^p] &=& \sum_{\alpha\in\S_p}(nk)^{\# \alpha},
\end{eqnarray*}
where $\gamma= (p \; p-1 \; \cdots\; 2 \; 1) \in\S_p$ is the full
cycle. In the second formula above, we recognize the generating
polynomial for the number of cycles of a permutation of $p$ objects
evaluated at $nk$. This is known to be equal to $nk(nk+1)(nk+2) \cdots
(nk+p-1)$ (see~\cite{stanley}, Proposition 1.3.4), and we obtain the
following theorem.
\begin{theorem}
%
%
%e12 ###
\begin{equation}\label{eq:moments_rank_one}
\E[\trace(Z^p)] = \Biggl( \prod_{j=0}^{p-1}(nk+j) \Biggr)^{-1}\sum
_{\alpha\in\S_p}k^{\# \alpha} n^{\#(\gamma^{-1}\alpha)}.
\end{equation}
\end{theorem}

This is exactly like formula (10) from~\cite{collins-nechita-1}, which
was obtained via the Weingarten formula. The approach followed here is
more straightforward and does not use unitary integration: It is based
on the purely combinatorial Wick formula and the Gaussianization trick.

%s4.3 ###
\subsection{Asymptotics}

We now look at the probability distribution of the output random matrix
$Z$ when one (or both) of the parameters $n$ and $k$ grow to infinity.
The asymptotic behavior of random matrices has been one of the main
objects of study in random matrix theory. For instance, it is in this
large-dimension regime that the freeness phenomenon appears. In the
particular case of random quantum channels under study here, this
question has an interesting physical motivation: large-dimensional
Hilbert spaces model physical systems with large numbers of degrees of
freedom. This point of view has been discussed in the quantum
information theory literature (see~\cite{braunstein,zycsommers,nechita,zycbook}).
Although some of what follows has already been
treated in~\cite{nechita}, the approach of this paper has the merit of
being self-contained and illustrates perfectly the power and range of
the Gaussian graphical calculus.

We split the results according to three possible asymptotic regimes,
depending on which of the parameters $n$ and/or $k$ is large. Of
special interest is the third regime, when \textit{both} parameters
grow to infinity,\vadjust{\goodbreak} but at a constant positive ratio $c>0$. We use the
equivalence symbol $x(n) \sim y(n)$ for nonzero sequences $x(n)$ and
$y(n)$ which are such that $x(n)/y(n) \to1$ when $n \to\iy$.

\begin{theorem}\label{thm:rank-one}
Let $Z = \Phi^{U,Y}(X)$ denote the output of a random quantum channel
$\Phi$, where $X$ and $Y$ are rank-one projectors.
\begin{enumerate}[(III)]
\item[(I)] In the regime where $n$ is fixed and $k \to\iy$, the
limiting spectral distribution of $Z$ is almost surely $\delta_{1/n}$.
\item[(II)] In the regime where $k$ is fixed and $n \to\iy$, $Z$
tends almost surely to a variable that has eigenvalues
$1/k$ with multiplicity $k$ and $0$ with multiplicity $n-k$.
\item[(III)] In the regime where $n,k \to\iy$, $k/n \to c >0$, $cnZ$
converges almost surely to a
free Poisson distribution with parameter $c$.
\end{enumerate}
\end{theorem}

\begin{pf}
In the first regime,
\[
\E[\trace(Z^p)] \stackrel{k \to\iy}{\sim} \frac{1}{n} (nk)^{-p}
\sum_{\alpha\in\S_p}k^{\# \alpha} n^{\#(\gamma^{-1}\alpha)}.
\]
Permutations $\alpha$ which give nonvanishing contributions are those
such that $\# \alpha= p$, hence $\alpha= \id$. In the end, we obtain
\[
\lim_{k \to\iy} \E[\trace(Z^p)]= n^{1-p},
\]
hence the limiting spectral distribution of $Z$ is $\delta_{1/n}$.

In order to prove the almost sure convergence,
we show that the empirical measures
\[
\mu_{n,k}(Z) = \frac{1}{n} \sum_{i=1}^n \lambda_i(Z)
\]
converge almost surely to the limit $\delta_{1/n}$ (which is
equivalent to the fact that, almost surely, every eigenvalue of $Z$
converges to $1/n$---recall that $n$ is fixed). As usual, almost sure
convergence of moments suffices and we aim to prove that for all $p,$
\[
\mbox{a.s.} \qquad  \lim_{k \to\iy} \trace(Z^p)= n^{1-p}.
\]
A standard application of Chebyshev's inequality and the
Borel--Cantelli lemma shows that it is enough to verify that for all
integers $p,$ the series of variances is summable:
\[
\sum_{k=1}^\iy\E\bigl[ \bigl(\trace(Z^p) - \E\trace(Z^p)
\bigr)^2 \bigr] < \iy.
\]
Let us separately compute $\E[\trace(Z^p)^2]$ and $\E[\trace
(Z^p)]^2$ using formula~(\ref{eq:moments_rank_one}). For the first
expectation, we need to introduce the permutation
%
%
%e13 ###
\begin{equation}\label{eq:gamma_2}
\gamma_2 = \bigl(p \; (p-1) \; \cdots\; 2 \; 1\bigr)\bigl(2p \; (2p-1) \;
\cdots\; (p+2) \; (p+1) \bigr) \in\S_{2p}.\vadjust{\goodbreak}
\end{equation}
We then have
\begin{eqnarray*}
\E[\trace(Z^p)^2] &=& \Biggl(\prod_{j=0}^{2p-1}(nk+j) \Biggr)^{-1}
\sum_{\alpha\in\S_{2p}} k^{\# \alpha} n^{\#(\gamma_2^{-1}\alpha
)}\\[-2pt]
&=& \Biggl(\prod_{j=0}^{2p-1} \biggl(1+\frac{j}{nk} \biggr)
\Biggr)^{-1} \sum_{\alpha\in\S_{2p}} k^{-|\alpha|} n^{-|\gamma
_2^{-1}\alpha|}.
\end{eqnarray*}
The first contribution (of order $k^0$) in the last sum is given by
$\alpha= \id$ and is equal to $n^{2-2p}$ (recall that $\gamma_2$ has
two cycles). The second-order in $k$ is given by transpositions $\alpha=
(ij)$. In this case, $|\gamma_2^{-1}\alpha| = 2p-3$ if $i$ and $j$
belong to the same cycle of $\gamma_2$ and $|\gamma_2^{-1}\alpha| =
2p-1$ otherwise. Hence, we obtain
\begin{eqnarray*}
\E[\trace(Z^p)^2] &=& \biggl[1 - \frac{2p(2p-1)}{2nk} + O \biggl(\frac
{1}{k^2} \biggr) \biggr] \\[-2pt]
&&{}\times
\biggl[ n^{2-2p} + \frac{1}{k}
\bigl(p^2 n^{1-2p} + p(p-1)n^{3-2p} \bigr) + O \biggl(\frac{1}{k^2}
\biggr) \biggr]\\[-2pt]
&=& n^{2-2p} + \frac{1}{k}p(p-1)n^{1-2p}(n^2-1) + O \biggl(\frac
{1}{k^2} \biggr).
\end{eqnarray*}

Using the same ideas, $\E[\trace(Z^p)]^2$ is easily computed:
\begin{eqnarray*}
&&\E[\trace(Z^p)]^2\\[-2pt]
&& \qquad  = \Biggl(\prod_{j=0}^{2p-1}(nk+j) \Biggr)^{-2}
\biggl( \sum_{\alpha\in\S_{p}} k^{\# \alpha} n^{\#(\gamma
^{-1}\alpha)} \biggr)^2\\[-2pt]
&& \qquad = \biggl[1 - \frac{p(p-1)}{2nk} + O \biggl(\frac{1}{k^2} \biggr)
\biggr]^2 \cdot\biggl[n^{1-p} + \frac{1}{k} \frac
{p(p-1)}{2}n^{2-p}+ O \biggl(\frac{1}{k^2} \biggr) \biggr]^2\\[-2pt]
&& \qquad = n^{2-2p} + \frac{1}{k}p(p-1)n^{1-2p}(n^2-1) + O \biggl(\frac
{1}{k^2} \biggr),
\end{eqnarray*}
and we conclude that $\E[\trace(Z^p)^2 - \E[\trace(Z^p)]^2 =
O(k^{-2})$. Thus, the covariance series converges, completing the proof.

In the second regime,
\[
\E[\trace(Z^p)] \stackrel{n \to\iy}{\sim} \sum_{\alpha\in\S
_p}k^{-|\alpha|} n^{-|\gamma^{-1}\alpha|}.
\]
The nonvanishing contribution is given by $\alpha= \gamma$ and thus
\[
\lim_{n \to\iy}\E[\trace(Z^p)] = k^{1-p}.
\]
In other words, for large $n$, $Z$ has the following eigenvalues:
\begin{itemize}
\item$1/k$ with multiplicity $k$;
\item$0$ with multiplicity $n-k$.\vadjust{\goodbreak}
\end{itemize}
The proof of the almost sure convergence follows the same lines as in
the previous case and is left to the reader.

In the third regime, after making the substitution $k=cn$, the
asymptotics are
%
%
%e14 ###
\begin{equation}\label{eq:Z_gaussian_III}
\E[\trace(Z^p)] \sim n^{-2p}c^{-p} \sum_{\alpha\in\S_p}c^{\#
\alpha} n^{\# \alpha+ \#(\gamma^{-1}\alpha)}.
\end{equation}
Since
%
%
%e15 ###
\begin{equation}\label{eq:ineq_rank_one}
\# \alpha+ \#(\gamma^{-1}\alpha) = 2p - (|\alpha| + |\gamma
^{-1}\alpha|) \leq p+1,
\end{equation}
we should rescale the matrix $Z$ by a factor of $n$. In fact, in order
to avoid some unnecessary complications, we shall rescale $Z$ by $c n$.
We get
\[
\E[\tracenorm_n((cnZ)^p)] \sim n^{-p-1} \sum_{\alpha\in\S
_p}c^{\# \alpha} n^{\# \alpha+ \#(\gamma^{-1}\alpha)}.
\]
Contributing permutations are those for which we have equality in
equation~(\ref{eq:ineq_rank_one}), that is, $|\alpha| + |\gamma
^{-1}\alpha| = |\gamma| = p-1$. These are permutations on the
geodesic $\id\to\gamma$ and are known to be in bijection with
noncrossing partitions $\sigma\in \mathit{NC}(p)$. Thus,
\[
\E[\tracenorm_n((cnZ)^p)] \sim\sum_{\sigma\in \mathit{NC}(p)}c^{\#\sigma}.
\]
One recognizes the moment--cumulant formula from free probability
theory. Hence, the limiting distribution of $cnZ$ has cumulants of all
orders equal to $c$ and we identify the free Poisson distribution of
parameter $c$. Let us now show that almost sure convergence holds:
\[
\lim_{n \to\iy}\E[\tracenorm_n((cnZ)^p)] = \sum_{\sigma\in
\mathit{NC}(p)}c^{\#\sigma} \qquad  \mbox{almost surely.}
\]
Using the same classical technique as in the first regime, we show that
the series
\[
\sum_n \bigl(\E[(\tracenorm_n((cnZ)^p)^2)] - \E[\tracenorm
_n((cnZ)^p)]^2 \bigr)
\]
converges.
We start by evaluating $\E[(\tracenorm_n((cnZ)^p)^2)]$ up to the
second-order in $n$. Using the permutation $\gamma_2$ defined in (\ref
{eq:gamma_2}) and the Gaussian graphical calculus, we have
\[
\E[(\tracenorm_n((cnZ)^p)^2)] \sim\sum_{\alpha\in\S_{2p}}c^{\#
\alpha} n^{2p-2-(|\alpha| + |\gamma_2^{-1}\alpha|)}.
\]
Using similar ideas as before, $|\alpha| + |\gamma_2^{-1}\alpha|
\geq|\gamma_2| = 2p-2$ with equality iff $\alpha$ is on the geodesic
between $\id$ and $\gamma_2$. Given the 2-cycle structure of $\gamma
_2$, geodesic permutations $\alpha$ admit a decomposition $\alpha=
\alpha' + \alpha''$, where $\alpha' \in\S\{1, 2, \ldots, p\} = \S
_p$ and\vadjust{\goodbreak} $\alpha'' \in\S\{p+1, p+2, \ldots, 2p\} \isom\S_p$ are
themselves geodesic permutations $\id\to\alpha' \to\gamma$ and
$\id\to\alpha'' \to\gamma$, respectively. Of course, in this case,
$\#\alpha= \# \alpha' + \# \alpha''$ and thus
\[
\E[(\tracenorm_n((cnZ)^p)^2)] \sim\mathop{\mathop{\sum}_{\id\to
\alpha'
\to\gamma}}_{ \id\to\alpha'' \to\gamma}c^{\# \alpha' + \#
\alpha
'} = \biggl( \sum_{\id\to\tilde\alpha\to\gamma}c^{\# \tilde
\alpha} \biggr)^2.
\]
By a standard parity argument, the function $\S_{2p} \ni\alpha
\mapsto(|\alpha| + |\gamma_2^{-1}\alpha|)\operatorname{mod} 2$ is constant and
thus there is no $n^{-1}$ term in the asymptotic development of $\E
[(\tracenorm_n((cnZ)^p)^2)]$:
\[
\E[(\tracenorm_n((cnZ)^p)^2)] = \biggl( \sum_{\id\to\tilde\alpha
\to\gamma}c^{\# \tilde\alpha} \biggr)^2 + O(n^{-2}).
\]
Similar ideas applied to   formula~(\ref{eq:Z_gaussian_III}) yield
the same conclusion:
\[
\E[\tracenorm_n((cnZ)^p)] = \sum_{\id\to\alpha\to\gamma}c^{\#
\alpha} + O(n^{-2}).
\]
Taking the square of this last equation and comparing it with the
previous one, we conclude that the general term of the covariance
series behaves asymptotically as $O(n^{-2})$. This implies that the
series is convergent and we conclude that the almost sure convergence holds.
\end{pf}

Even though Gaussianization results are exact and do not require a
detour through Weingarten calculus, it is not clear how to apply them
when the input is not one-dimensional. However, it is natural to wonder
about the asymptotics in this case as well. The calculus that we
introduced in~\cite{collins-nechita-1} is crucial for this and that is
the subject of Section~\ref{sec:general_inputs}.% '

%s4.4 ###
\subsection{Almost sure convergence for entropies}

In this section, we improve the almost sure convergence of moments to
the almost sure convergence
of any continuous function with polynomial growth. Since the set of
functions that it applies to is larger, this type of convergence is
stronger than the weak convergence.
We deduce corollaries for quantum information theory, and the
techniques developed in
this section will be useful toward the end of the paper. The technique
of proof for this result is inspired by~\cite{guionnet}.

\begin{proposition}\label{strengthened-weak}
Let $f$ be a continuous function on $\R$ with polynomial growth and
$\nu_n$ be a sequence of probability measures which converges
in moments to a compactly supported measure $\nu$. Then,
$\int f \,d\nu_n \to\int f \,d\nu$.
\end{proposition}

\begin{pf}
Let $K>1$ be a constant such that the interval $[-(K-1), K-1]$ contains
the (compact) support of the limit measure $\nu$. It follows that, for
all integer powers $s \geq0$,
%
%
%e16 ###
\begin{equation}\label{eq:Cheby}
\lim_{r \to\iy} K^{-2r} \int x^{2r+2s} \,d\nu(x) = 0.\vadjust{\goodbreak}
\end{equation}
Moreover, since the measures $\nu_n$ converge in moments to $\nu$,
for all $\varepsilon>0$, there exists an $r$ large enough such that
for all $n$ large enough,
%
%
%e17 ###
\begin{equation}\label{eq:Cheby-n}
K^{-2r} \int x^{2r+2s} \,d\nu_n(x) <\varepsilon.\vspace*{-1pt}
\end{equation}
For some fixed $\delta> 0$, the Weierstrass theorem produces a
polynomial $P$ such that
$|f(x) - P(x)|<\delta$ for all $x \in[-K, K]$. We then have
\[
\biggl|\int f \,d\nu_n - \int f \,d\nu\biggr| \leq\int|f-P| \,d\nu_n + \int
|f-P| \,d\nu+ \biggl| \int P \,d(\nu_n-\nu) \biggr|.\vspace*{-1pt}
\]
Since the polynomial approximation holds on the support of $\nu$, the
second term above is less than $\delta$. Using the convergence in
moments of the probability measures $\nu_n$, the last term can be seen
to be less than $\delta$ for $n$ large enough. We focus now on the
first term above, $\int|f-P| \,d\nu_n$. By the polynomial
approximation, $\int|f-P| \,d\nu_n \leq\delta+ \int_{|x|\geq K}
|f-P| \,d\nu_n$. Since $f$ has polynomial growth, one can find a
constant $q >0$ such that $|f(x)-P(x)| \leq x^{2q}$ for all $|x| \geq
K$. Using the Chebyshev inequality on the last integral, we have, for
all $r \geq1,$
\[
\int_{|x|\geq K} |f-P| \leq\int_\R\frac{x^{2r}}{K^{2r}} x^{2q} \,
d\nu_n = K^{-2r}\int x^{2q+2r} \,d\nu_n.\vspace*{-1pt}
\]
The convergence in moments, together with equations~\eqref{eq:Cheby}
and~\eqref{eq:Cheby-n}, implies that, for $r$ and $n$ large enough,
the above expression can be made arbitrarily small, which completes the
proof.\vspace*{-2pt}
\end{pf}

\begin{remark}
The conclusion of the above proposition still holds true under the
weaker assumption that $\nu$ admits some finite exponential moment,
thanks to the fact that weak convergence of measures implies
convergence of integrals under uniform integrability. However, in this
paper we only consider compactly supported measures.\vspace*{-2pt}
\end{remark}

\begin{corollary}
\label{cor-entropy}
Almost surely, in the regime $n \to\iy$, $k \sim cn$, the von Neumann
entropy of the matrix $Z$ from Theorem~\ref{thm:rank-one} satisfies
\[
H(Z) =
\cases{\displaystyle
\log n - \frac{1}{2c}+o(1), & \quad   if   $c\geq1$,\cr\displaystyle
\log(cn) - \frac{c}{2}+o(1), & \quad    if   $0<c<1$.
}\vspace*{-2pt}
\]
\end{corollary}

\begin{pf}
Let us assume that $c \geq1$, the other case being similar. We use
Theorem~\ref{strengthened-weak}
for the function $x \mapsto x\log x,$ which is continuous and of
polynomial growth on the domain $\R_+$, and for the empirical spectral
measures of the matrices $cnZ$. It follows that, almost surely when $n
\to\iy$,
\[
\frac{1}{n}\sum_{i=1}^n cn\lambda_i\log(cn\lambda_i)=\int t\log t
\,d\pi_c (t)+o(1),\vspace*{-1pt}
\]
where $\lambda_1 \geq\cdots\geq\lambda_n$ are the eigenvalues of $Z$.\vadjust{\goodbreak}

Simplifying this expression and using the value of the right-hand side
integral from equation~\eqref{eq:entropy-free-poisson}, we have
\[
H(Z) = -\sum_{i=1}^n \lambda_i\log\lambda_i=\log n - \frac{1}{2c}+o(1),
\]
completing the proof.
\end{pf}

A formula of Page~\cite{page} states that the mean entropy of a random
density matrix $Z^{(n,k)} \in\M_n(\C)$ obtained by tracing out a
$k$-dimensional environment is given by (here, $n \leq k$ are fixed)
\[
\E H\bigl(Z^{(n,k)}\bigr) = \sum_{j=k+1}^{nk} \frac{1}{j} - \frac{n-1}{2k}.
\]
We could obtain a weaker version of Corollary~\ref{cor-entropy} from
Page's formula by letting $n$ tend to infinity and using the dominated
convergence theorem.

%s5 ###
\section{Asymptotics of a single random quantum channel for general
states}\label{sec:general_inputs}

%s5.1 ###
\subsection{The model}

We are interested in single random quantum channels and
study the asymptotic behavior of the output of such channels for more
general input states than rank-one projectors.
The Gaussian planar expansion cannot be used in the more general cases,
so we need the Weingarten planar expansion.
We may consider the general model
%
%
%e18 ###
\begin{equation}\label{eq:trace_scaling}
\trace_\beta(X) \sim(n^s)^{\#\beta}u^{\#\beta}\phi_\beta(x),
\end{equation}
where $s,u \in\R$ are fixed parameters and $x$ is a random variable
in some noncommutative probability space with trace $\phi$.
In this section, we will deal only with two special cases of interest
of the above formula. The first is motivated by quantum information
theory: $X$ is a rank-$r$ projector. This choice corresponds to $s=0$,
$u=r$ and $x=r^{-1}$. The second special case we consider will seem
natural to the reader with a free probabilistic background: $X$
converges in moments to a noncommutative random variable $x$. To obtain
this particular case from formula~(\ref{eq:trace_scaling}), we have to
put $s=u=1$ (this amounts to taking a normalized trace in the left-hand
side). Note, however, that such an input matrix is not normalized, and
we have to take into account the trace-one restriction for quantum states.

Let us recall here the formula for the moments of the output $Z=\Phi
(X)$ of a random quantum channel (see~\cite{collins-nechita-1}):
%
%
%e19 ###
\begin{equation}\label{eq:moments_Z}
\E[\trace(Z^p)] = \sum_{\alpha, \beta\in\S_p} k^{\# \alpha}
n^{\#(\gamma^{-1} \alpha)} \trace_\beta(X) \Wg(\alpha\beta^{-1}),
\end{equation}
where $\gamma$ is the full-cycle permutation $\gamma= (p \; p-1 \;
\cdots\; 2 \; 1) \in\S_p$.

%s5.2 ###
\subsection{Rank-$r$ projectors}
Plugging, for all $\beta\in\S_p$, $\trace_\beta(X) = r^{\#\beta
}r^{-p} = r^{-|\beta|}$ into the previous equation, we obtain
%
%
%e20 ###
\begin{equation}\label{eq:moments_Z_rank_r}
\E[\trace(Z^p)] = \sum_{\alpha, \beta\in\S_p} k^{\# \alpha}
n^{\#(\gamma^{-1} \alpha)} r^{-|\beta|} \Wg(\alpha\beta^{-1}).
\end{equation}

We study, as usual, the three following asymptotic regimes: $n$ fixed,
$k \to\iy$; $k$~fixed, $n \to\iy$; $n,k \to\iy$, $k/n \to c$.

\begin{proposition}
Depending on the asymptotic regime, the almost sure behavior of $Z$ is
given as follows:
\begin{enumerate}[(III)]
\item[(I)] when $n$ is fixed and $k \to\iy$, the output density
matrix $Z$ converges almost surely to the maximally mixed state
\[
\rho_* = \frac{1}{n}\I_n;
\]
\item[(II)] when $k$ is fixed and $n \to\iy$, the output density
matrix $Z$, restricted to its support of dimension $rk,$ converges to
$1/(rk) \I_{rk}$;
\item[(III)] finally, in the third regime $k/n \to c$, the empirical
spectral distribution~of the matrix $rkZ$ converges to a free Poisson
distribution of parameter~$rc$.
\end{enumerate}
\end{proposition}

\begin{pf}
Using the Weingarten asymptotic $\Wg(\alpha\beta^{-1}) \sim
(nk)^{-p-|\alpha\beta^{-1}|}$, the exponent of $k$ in equation (\ref
{eq:moments_Z_rank_r}) is given by $\# \alpha- p - |\alpha\beta
^{-1}|$. This reaches its maximum of zero when $\alpha=\beta=\id$.
Hence, to the first-order in $k$, we have
\[
\E[\trace(Z^p)] = n^{1-p} + o(1),
\]
and the conclusion follows.

The second regime is very similar, and we ultimately obtain (this time
up to the first order in $n$)
\[
\E[\trace(Z^p)] = (rk)^{1-p} + o(1).
\]

As for the third regime, making the substitution $k=cn$, we obtain the
following asymptotic relation:
\[
\E[\trace(Z^p)] \sim\sum_{\alpha, \beta\in\S_p} r^{-|\beta|}
c^{-(|\alpha| + |\alpha\beta^{-1}|)} n^{-(|\alpha| + |\gamma^{-1}
\alpha| + 2|\alpha\beta^{-1}|)} \Mob(\alpha\beta^{-1}).
\]
The exponent of the large parameter $n$ in the last formula is
minimized when $\id\to\alpha=\beta\to\gamma$ is a geodesic in $\S
_p$. Hence,
\[
\E[\trace(Z^p)] \sim n^{1-p}\sum_{\id\to\alpha\to\gamma}
(rc)^{-|\beta|} \Mob(\alpha\beta^{-1}).
\]
Thus, the normalized trace of the $p$th power of the matrix $rkZ$
converges to
\[
\sum_{\id\to\alpha\to\gamma} (rc)^{\# \beta} = \sum_{\sigma\in
\mathit{NC}(p)} (rc)^{\# \sigma}\vadjust{\goodbreak}
\]
and we easily recognize the moment--cumulant formula for the
Marchenko--Pastur distribution of parameter $rc$
(see Section
\ref{reminders-wishart}).

The above results have been proven to hold for the convergence in moments.
Borel--Cantelli techniques (see~\cite{collins-nechita-1} for a sample)
can be easily used to show that the stronger almost sure convergence
holds in all three cases.
\end{pf}

%s5.3 ###
\subsection{Normalized macroscopic inputs}

We now consider matrices $X$ which have a macroscopic scaling $\trace
(X^p) \sim n \cdot\phi(x^p)$,
where $x$ is some noncommutative random variable. We have, of course,
to normalize such input matrices and we shall consider
\[
\tilde X = \frac{X}{\trace X}.
\]
With this normalization, the moments of the output matrix $Z = \Phi
(\tilde X)$ are given by
\[
\E[\trace(Z^p)] = \E[\trace(\Phi(\tilde X)^p)] = \E\biggl[ \trace
\frac{\Phi(X)^p}{(\trace X)^p} \biggr] = \frac{\E[\trace(\Phi
(X)^p)]}{(\trace X)^p}.
\]

As in the previous section, we consider different asymptotic regimes
for the integer parameters $n$ and $k$.
However, it turns out that the $k$ fixed, $n \to\iy$ regime is more
involved, and its understanding requires some more advanced free
probabilistic tools. To an integer $k$ and a probability measure~$\mu
$, we associate the measure $\mu_{(k)}$ defined by
\[
\mu_{(k)} = \biggl( 1 - \frac1 k \biggr) \delta_0 + \frac1 k \mu.
\]

\begin{proposition}
The almost sure behavior of the output matrix $Z=\Phi(\tilde X)$ is
given as follows:
\begin{enumerate}[(III)]
\item[(I)] when $n$ is fixed and $k \to\iy$, $Z$ converges almost
surely to the maximally mixed state
\[
\rho_* = \frac{1}{n}\I_n;
\]
\item[(II)] when $k$ is fixed and $n \to\iy$, the empirical spectral
distribution of $\bar\mu k n Z$ converges to the probability measure
$\nu= [\mu_{(k)}]^{\boxplus k^2}$, where $\boxplus$ denotes the free
additive convolution operation, $\mu$ is the probability distribution
of $x$ with respect to $\phi$: $\phi(x^p) = \int t^p \,d\mu(t)$ and
$\bar\mu$ is the mean of $\mu$, $\bar\mu= \phi(x)$;
\item[(III)] when $n,k \to\iy$ and $k/n \to c$, the empirical
spectral distribution of the matrix $nZ$ converges to the Dirac mass
$\delta_1$.
\end{enumerate}
\end{proposition}

\begin{pf}
We start with the simplest asymptotic regime, $n$ fixed and $k \to\iy
$. Plugging the scaling for $\trace_\beta(X)$ into formula (\ref
{eq:moments_Z}), we get
\[
\E[\trace(Z^p)] \sim n^{-p} \phi(x)^{-p}\!\sum_{\alpha, \beta\in\S
_p}\!\! k^{\# \alpha} n^{\#(\gamma^{-1} \alpha)} n^{\#\beta} \phi
_\beta(x) (nk)^{-p-|\alpha\beta^{-1}|} \Mob(\alpha\beta^{-1}).\vadjust{\goodbreak}
\]
In order to find the leading term in the preceding sum, we have to
minimize the exponent of $k$, $|\alpha|+|\alpha\beta^{-1}|$. This
expression attains its minimum 0 at $\alpha=\beta=\id$. In the end,
we find $\E[\trace(Z^p)] \sim n^{1-p}$ and conclude that the output
matrix $Z$ converges to the maximally mixed state $\rho_* = \I_n / n$.

Let us now look at the second regime, $k$ fixed and $n \to\iy$. The
asymptotic moments of $Z$ are given by
\[
\E[\trace(Z^p)]\sim\phi(x)^{-p}\!\sum_{\alpha, \beta\in\S_p}\!\!
k^{-(|\alpha| + |\alpha\beta^{-1}|)}n^{-(|\beta| + |\alpha\beta
^{-1}| + |\gamma^{-1} \alpha|)} \phi_\beta(x) \Mob(\alpha\beta^{-1}).
\]
The dominating terms in the preceding sum are given by permutations
such that $|\beta| + |\alpha\beta^{-1}| + |\gamma^{-1} \alpha|$ is
minimal. Permutations $(\alpha, \beta)$ which saturate the triangle
inequality $|\beta| + |\alpha\beta^{-1}| + |\gamma^{-1} \alpha|
\geq|\gamma| = p-1$ are elements of the geodesic $\id\to\beta\to
\alpha\to\gamma$ and can be put in bijection with noncrossing
partitions $\sigma\leq\tau\in \mathit{NC}(p)$ using Lemma~\ref{lem:S_p}. We obtain
\[
\frac{1}{n}\E[\trace((\bar\mu k n Z)^p)]\sim\sum_{\sigma\leq
\tau\in \mathit{NC}(p)} k^{2\#\tau- \#\sigma} \phi_\sigma(x) \Mob(\sigma,
\tau) .
\]
Using the fact that $k^{-\#\sigma} \phi_\sigma(x) = \phi_\sigma
(\mu_{(k)})$ and applying the moment--cumulant formula (\cite
{nica-speicher}, page 175), we get
\begin{eqnarray*}
\frac{1}{n}\E[\trace((\bar\mu k n Z)^p)]&\sim&\sum_{ \tau\in
\mathit{NC}(p)} k^{2\#\tau} \mathop{\mathop{\sum}_{\sigma\in \mathit{NC}(p) }}_{
\sigma\leq
\tau}\phi_\sigma\bigl(\mu_{(k)}\bigr) \Mob(\sigma, \tau)\\
 &=& \sum_{ \tau
\in \mathit{NC}(p)} k^{2\#\tau} \kappa_\tau\bigl(\mu_{(k)}\bigr) ,
\end{eqnarray*}
where $\kappa$ denotes the free cumulant. We conclude that the random
matrix $\bar\mu k n Z$ converges in distribution to a probability
measure $\nu$ which has free cumulants $\kappa_p(\nu) = k^2 \kappa
_p(\mu_{(k)})$, and the conclusion follows.

We now turn to the third regime, where both $n$ and $k$ grow to
infinity at a constant ratio $c>0$. After making the substitution
$k=cn$, we obtain the following equivalent:
\begin{eqnarray*}
&&\E[\trace(Z^p)]\\
&& \qquad  \sim\phi(x)^{-p}\sum_{\alpha, \beta\in\S_p}
n^{-(|\alpha|+|\gamma^{-1}\alpha|+|\beta|+2|\alpha\beta^{-1}|)}
c^{-(|\alpha|+|\alpha\beta^{-1}|)}\phi_\beta(x) \Mob(\alpha\beta^{-1}).
\end{eqnarray*}
The expression to minimize in this case is $|\alpha|+|\gamma
^{-1}\alpha|+|\beta|+2|\alpha\beta^{-1}|$. By the triangle inequality,
(cf. Lemma~\ref{lem:S_p}),
the sum of the first two terms is at least $|\gamma| = p-1$ and the
other terms are positive; hence, the (negative) exponent of $n$ is at
least $p-1$, and the bound is reached for $\alpha=\beta=\id$. To the
first-order in $n$, the asymptotic moments of $Z$ are
\[
\E[\trace(Z^p)] \sim n^{1-p} \qquad  \forall p \geq1,
\]
which is equivalent to the statement
\[
\lim_{n \to\iy}\E[\tracenorm_n((nZ)^p)] = 1 \qquad  \forall p \geq1.
\]

In all the cases treated above, we leave the proof of the almost sure
convergence to the reader.
\end{pf}

\begin{remark}
Let us observe that for the regimes (I) and (III) studied above,
the limit distribution of the output does not depend on the limit of
the input distribution. The result obtained in the second regime could
have been obtained in a more direct manner, using the powerful tools of
free probability. For simplicity, let us forget about the normalization
of the input matrix and observe that the limit distribution of $X
\otimes Y$ is $\mu_{(k)}$, if $\mu$ is the limit distribution of $X$
and $Y$ is a $k \times k$ rank-one projector. The partial trace of the
randomly rotated input matrix is equal to the sum of its $k$ $n\times
n$ diagonal blocks. Each block is a free compression of parameter $1/k$
(which accounts for a free additive convolution power of $k$), and the
blocks are free. Taking the sum of the free blocks explains the other
factor $k$ appearing as an exponent for the free additive convolution.
\end{remark}

%s6 ###
\section{Tensor products of quantum channels}

%s6.1 ###
\subsection{Motivation and existing results}\label{sec:product_channels_intro}
When studying the question of the additivity of minimal output
entropies, it is natural to consider products of random quantum channels.

Before looking in detail at some specific models, let us observe that
if one chooses an input state which factorizes $X_{12} = X_1 \otimes
X_2$, then
\[
[\Phi_1 \otimes\Phi_2](X_{12}) = \Phi_1(X_1) \otimes\Phi_2(X_2),
\]
and there is no correlation (classical or quantum) between the
channels. In order to avoid such trivial situations, we must choose an
input state which is entangled. An obvious choice (given that $n_1 =
n_2 = n$) is to take $X_{12} = E_n$, the $n$-dimensional Bell state,
and we shall use this state in what follows.

 Winter and  Hayden observed in~\cite{hayden-winter} that it is
relevant in this framework to introduce the further symmetry $U_2 = \ol
U_1$, as it ensures that at least one eigenvalue is always large.
In~\cite{collins-nechita-1}, using the channel model inspired by the
ideas of Hayden and Winter, it was proven that the bounds on the
eigenvalues could be improved as follows.

\begin{theorem}
\label{thm:bell-phenomenon}
In the $k$ fixed, $n \to\iy$ regime, the eigenvalues of the matrix
$Z$ converge \textit{almost surely} toward:
\begin{itemize}
\item$\frac{1}{k} + \frac{1}{k^2} - \frac{1}{k^3}$, with
multiplicity one;
\item$\frac{1}{k^2} - \frac{1}{k^3}$, with multiplicity $k^2-1$;
\item$0$, with multiplicity $n^2 - k^2$.\vadjust{\goodbreak}
\end{itemize}

In the asymptotic regime where $n$ is fixed and $k \to\iy$, the
random matrix~$Z$ converges to the chaotic state
\[
\rho_* = \frac{\I_{n^2}}{n^2}.
\]
\end{theorem}

If we look for optimal bounds for the minimum output entropy of $\Phi
\otimes\overline{\Phi}$, then there is no mathematical proof that
$U_2 = \ol U_1$ is the best choice. Actually, this choice of
probability measure on $\U(n)\times\U(n)$
does not have full support and we cannot rule out that the maximum for
the minimum output entropy is outside the support of the probability measure.
This is what motivates the introduction of the example in which $U_1$
and $U_2$ are independent unitary matrices.
As we will see, this does not yield improvements on the example of
Winter with high probability.
More strikingly, in the regimes that we consider, we will see that the
constraint $U_2 = \ol U_1$ yields no significant
improvement to the asymptotic behavior of the von Neumann entropies,
and this suggests that the simpler random model where $U_1$ and $U_2$
are independent could be a candidate for additivity violation with high
probability.

In the forthcoming subsections we analyze both models (independent and
conjugate unitaries) in a different asymptotic regime, where both
parameters $n$ and $k$ grow to infinity at a constant ratio $k/n \to
c$. The model where the quantum channels are independent has received
less attention from the quantum information community; here, we show
that it is intimately connected to the (more interesting) case of
conjugate channels, by comparing eigenvalue profiles for outputs of
channels from the two families.

%s6.2 ###
\subsection{Independent interaction unitaries}

Here, we consider two \textit{independent} realizations $U_1 = U$ and
$U_2 = V$ of Haar-distributed unitary random matrices on $\U(nk)$. For
both channels the state of the environment is a rank-one projector and
we are interested in the $n^2 \times n^2$ random matrix
\[
Z = [\Phi^U \otimes\Phi^V] (E_n),
\]
where $E_n$ is the maximal entangled Bell state
\[
E_n = \frac{1}{n} \sum_{i,j=1}^n | e_i \rangle\langle e_j | \otimes
| e_i \rangle\langle e_j |.
\]

The diagram associated with the $(2,2)$ tensor $Z$ appears in Figure \ref
{fig:bi_channel_UV}.

%
%f7 ###
\begin{figure}

\includegraphics{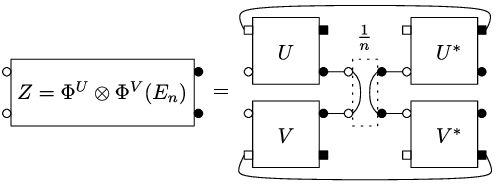}

\caption{$Z = \Phi^U \otimes\Phi^V (E_n).$}
\label{fig:bi_channel_UV}
\end{figure}

We compute the moments $\E[\trace(Z^p)]$ for all $p \geq1$ using the
graphical method. We start, as depicted in Figure \ref
{fig:bi_channel_UV}, by replacing $U^*$ (resp., $V^*$) boxes by $\bar U$
(resp., $\bar V$) boxes.
Notice that
there are two types of boxes corresponding to the independent random
unitary matrices $U$ and $V$ (when computing the $p$th moment of $Z$,
there are $p$ boxes of each type). This has two important consequences:
when expanding the diagram in order to compute the expectation\vadjust{\goodbreak} of the
trace, we can only pair $U$ boxes with $\bar U$ boxes and $V$ boxes
with $\bar V$ boxes; ``cross-pairings'' between ``$U$'' boxes and
``$V$'' boxes are not allowed by the expansion algorithm.
In addition, we have to index the Weingarten sum by two pairs of
permutations, one for each type of box (we shall denote them by $\alpha
_U, \beta_U, \alpha_V, \beta_V \in\S_p$). The four permutations
are responsible for pairing boxes in the following ways ($1 \leq i \leq p$):
\begin{longlist}[(3)]
\item[(1)] the inputs of the $i$th $U$-box are paired with the inputs of the
$\alpha_U(i)$th $\bar U$ box;
\item[(2)] the outputs of the $i$th $U$-box are paired with the outputs of
the $\beta_U(i)$th $\bar U$ box;
\item[(3)] the inputs of the $i$th $V$-box are paired with the inputs of the
$\alpha_V(i)$th $\bar V$ box;
\item[(4)] the outputs of the $i$th $V$-box are paired with the outputs of
the $\beta_V(i)$th $\bar V$ box.
\end{longlist}
Since our diagram consists only of unitary matrices (there are no
constant nontrivial tensors), the result of the graph expansion is a
(sum over a) collection of loops, multiplied by some scalar factor. The
different contributions of a general quadruple $(\alpha_U, \beta_U,
\alpha_V, \beta_V) \in\S_p^4 $ are given by (recall that circles
correspond to $n$-dimensional spaces and squares correspond to
$k$-dimensional spaces):
\begin{longlist}[(8)]
\item[(1)] loops from $\square U$ and $\bar U \square$: $k^{\# \alpha_U}$;
\item[(2)] loops from $\circ U$ and $\bar U \circ$: $n^{\#(\gamma
^{-1}\alpha_U)}$;
\item[(3)] loops from $U \blacksquare$ and $\blacksquare\bar U$: none;
\item[(4)] loops from $U \bullet$, $\bullet\bar U$, $V \bullet$ and
$\bullet\bar V$: $n^{\# (\beta_U^{-1}\beta_V)}$;
\item[(5)] loops from $\square V$ and $\bar V \square$: $k^{\# \alpha_V}$;
\item[(6)] loops from $\circ V$ and $\bar V \circ$: $n^{\#(\gamma
^{-1}\alpha_V)}$;
\item[(7)] normalization factors $1/n$ from the Bell matrices $E_n$: $n^{-p}$;\vspace*{1pt}
\item[(8)] Weingarten weights for the $U$-matrices: $\Wg(\alpha_U\beta_U^{-1})$;
\item[(9)] Weingarten weights for the $V$-matrices: $\Wg(\alpha_V\beta_V^{-1})$.
\end{longlist}
Adding all these contributions, we obtain an exact closed-form
expression, as follows.\vadjust{\goodbreak}

\begin{proposition}
The moments of the random variable $Z$ can be computed as follows:
%
%
%e21 ###
\begin{eqnarray}\label{eq:bi_canal_UV_det}
&&\E[\trace(Z^p)] = \sum_{\alpha_U, \beta
_U, \alpha_V, \beta_V \in\S_p} k^{\#
\alpha_U + \# \alpha_V}n^{\#(\gamma^{-1}\alpha_U) + \#(\gamma
^{-1}\alpha_V) + \# (\beta_U^{-1}\beta_V)-p}\nonumber
\\[-8pt]
\\[-8pt]
&&\hphantom{\E[\trace(Z^p)] = \sum_{\alpha_U, \beta
_U, \alpha_V, \beta_V \in\S_p}}
{}\times\Wg(\alpha_U\beta
_U^{-1})\Wg(\alpha_V\beta_V^{-1}).
\nonumber
\end{eqnarray}
\end{proposition}

Here, we study the asymptotic regime $n,k \to\iy$, $k/n \to c >0$.
Our main theorem is as follows.

\begin{theorem}
\label{thm:asympt-UV}
Almost surely, in the regime $n \to\iy$, $k \sim cn$, the
distribution of the output matrix $c^2n^2Z$ converges toward a free
Poisson law with parameter $c^2$.
\end{theorem}

\begin{pf}
We start by replacing $k$ by $c n$ in equation (\ref
{eq:bi_canal_UV_det}) and obtain
\[
\E[\trace(Z^p)] \sim\sum_{\alpha_U, \beta_U, \alpha_V, \beta_V
\in\S_p} n^{-\mathcal P_n} c^{-\mathcal P_c} \Mob(\alpha_U\beta
_U^{-1})\Mob(\alpha_V\beta_V^{-1}),
\]
where
\[
\mathcal P_n = |\alpha_U| + |\alpha_V| + |\gamma^{-1}\alpha_U| +
|\gamma^{-1}\alpha_V| + |\beta_U^{-1}\beta_V| + 2|\alpha_U\beta
_U^{-1}| + 2|\alpha_V\beta_V^{-1}|
\]
and
\[
\mathcal P_c = |\alpha_U| + |\alpha_V| + |\alpha_U\beta_U^{-1}| +
|\alpha_V\beta_V^{-1}|.
\]

Since we are interested in the asymptotic $n \to\iy$ ($c$ is a
constant), we want to minimize $\mathcal P_n$. The following
inequalities are standard
(cf. Lemma~\ref{lem:S_p}):
%
%e24 ###
%e23 ###
%e22 ###
\begin{eqnarray}
\label{eq:UV_ineq_first}|\alpha_U| + |\gamma^{-1}\alpha_U| &\geq&
p-1;\\
|\alpha_V| + |\gamma^{-1}\alpha_V| &\geq& p-1;\\
\label{eq:UV_ineq_last}
|\beta_U^{-1}\beta_V| , 2|\alpha_U\beta
_U^{-1}| , 2|\alpha_V\beta_V^{-1}| &\geq&0,
\end{eqnarray}
and thus $\mathcal P_n \geq2p-2$ with equality iff $\alpha_U = \beta
_U = \alpha_V = \beta_V = \alpha$ and $\alpha$ is on a geodesic
between $\id$ and $\gamma$. By choosing the obvious $n^2$ rescaling,
we get
\[
\lim_{n,k \to\iy}\E\biggl[\frac{1}{n^2}\trace((c^2n^2Z)^p)
\biggr] = \!\sum_{\alpha\mathrm{\ geodesic}}\! c^{2p-2|\alpha|} =\!
\sum_{\alpha
\mathrm{\ geodesic}}\! c^{2\# \alpha} =\! \sum_{\sigma\in \mathit{NC}(p)}\! c^{2\#
\sigma},
\]
and we recognize in the last sum the $p$th moment of the free Poisson
distribution with parameter $c^2$. This shows that the the matrix $Z$
converges \textit{in moments} to the limiting Marchenko--Pastur
distribution. The argument for the almost sure convergence relies on
the Borel--Cantelli lemma and can be found in the \hyperref[appm]{Appendix}.
\end{pf}

The von Neumann entropy of the output can be calculated in a fashion
similar to Corollary~\ref{cor-entropy}.

\begin{proposition}
\label{prop:entropy_independent}
Almost surely, in the limit $n \to\iy$, the von Neumann entropy of
the matrix $Z$ satisfies
\[
H(Z) =
\cases{\displaystyle
2\log n - \frac{1}{2c^2}+o(1), & \quad   if   $c\geq1$,\cr\displaystyle
2\log(cn) - \frac{c^2}{2}+o(1), & \quad   if   $0<c<1$.
}
\]
\end{proposition}

%
%f8 ###
\begin{figure}

\includegraphics{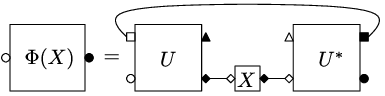}

\caption{A quantum channel with asymmetric input and output tensor structure.}
\label{fig:channel_d}
\end{figure}

Let us now consider a slightly generalized model of random quantum
channels. We introduce channels $\Phi\dvtx  \M_d(\C) \to\M_n(\C)$
which have a different tensor product structure at their input and
output. Here, $d$ is an integer parameter, and we shall always suppose
that $d | nk$. The diagram associated with such a channel is
depicted in Figure~\ref{fig:channel_d}, where diamond-shaped labels
correspond to $d$-dimensional vector spaces and triangle-shaped
decorations denote spaces of dimension $d'=nk/d$. The above analysis
for a product of independent channels is easily adapted to this more
general situation:\looseness=1
\[
\E[\trace(Z^p)] \sim\sum_{\alpha_U, \beta_U, \alpha_V, \beta_V
\in\S_p} n^{-\tilde{\mathcal P_n}} d^{-\tilde{\mathcal P_d}}
c^{-\tilde{\mathcal P_c}} \Mob(\alpha_U\beta_U^{-1})\Mob(\alpha
_V\beta_V^{-1}),
\]\looseness=0
where
\begin{eqnarray*}
\tilde{\mathcal P_n} &=& |\alpha_U| + |\alpha_V| + |\gamma^{-1}\alpha
_U| + |\gamma^{-1}\alpha_V| + 2|\alpha_U\beta_U^{-1}| + 2|\alpha
_V\beta_V^{-1}|,
\\
\tilde{\mathcal P_d} &=& |\beta_U^{-1}\beta_V|  \quad \mbox{and} \quad
\tilde{\mathcal P_c} = |\alpha_U| + |\alpha_V| + |\alpha
_U\beta_U^{-1}| + |\alpha_V\beta_V^{-1}|.
\end{eqnarray*}

\begin{remark}
If $d=d(n)$ is a function of $n$ such that $\lim_{n \to\iy} d(n) =
\iy$, then the considerations in Theorem~\ref{thm:asympt-UV} carry
over to this case and we obtain exactly the same limit, a~free Poisson
distribution with parameter $c^2$. The function $d=d(n)$ does not play
any role in this situation.
\end{remark}

On the other hand, if the parameter $d$ is constant (inputs of fixed
dimension), then the limiting behavior changes. Indeed, the minimizing
constraint $|\beta_U^{-1}\beta_V| = 0$ disappears, and the
contributing quadruples of permutations become uncoupled: $\id\to
\alpha_U=\beta_U \to\gamma$ and $\id\to\alpha_V=\beta_V \to
\gamma$. In conclusion, the asymptotic moments in this case are given
by the formula, which we summarize in the following proposition.
\begin{proposition}
If $d$ is constant, the limiting distribution of $c^2n^2Z$ also exists
and its limit moments are given by
\[
\frac{1}{n^2}\E[\trace((c^2n^2Z)^p)] \sim\mathop{\mathop{\sum
}_{\id\to
\alpha_U=\beta_U \to\gamma}}_{ \id\to\alpha_V=\beta_V \to
\gamma
} c^{\#\alpha_U + \#\alpha_V} d^{-|\alpha_U^{-1}\alpha_V|}.
\]
\end{proposition}

\begin{question}
We are not able to identify this distribution, even though its
properties look new. We wonder whether this distribution could be
related to generalized convolutions of Bo{\.z}ejko and coworkers
(cf.~\cite{bozejko}).
\end{question}

%s6.3 ###
\subsection{Conjugate interaction unitaries}
\label{subsec:conjugate}
To conclude, we consider the tensor product
of two \textit{conjugate} random quantum channels. As was emphasized
in Section~\ref{sec:product_channels_intro}, product channels $\Phi_U
\otimes\Phi_{\ol U}$ have very interesting eigenvalues statistics and
have received a lot of attention in the last years because of their
usefulness in providing counterexamples to different additivity conjectures.
The purpose of this section is to obtain a description of the behavior
of such channels in the regime where both $n$ and $k$ grow to infinity
at a constant ratio $c \in(0, \iy)$.

Hayden and Winter
remarked in~\cite{hayden-winter} that such a conjugate product channel
has a very important property: the output of the maximally entangled
state over the input space has a ``large'' eigenvalue, of size at least
$1/(cn)$.
The results of~\cite{collins-nechita-1} show that one expects for this
model a large eigenvalue $\lambda_1 = 1/(cn) + o(1/n)$ and $(n^2-1)$
smaller eigenvalues.
The purpose of this section is to show that this is indeed the case.
Actually, we can prove that
the random matrix under study has eigenvalues on two scalings: $1/n$
and $1/n^2$. In the next theorem, we compute the moments of the output
matrix $Z$ up to the first order in $n$.

\begin{theorem}
\label{thm:asympt-Z}
Fix some scaling constant $c>0$ and consider a sequence of
random quantum channels $\Phi_{n,k}$,
where $n,k \to\iy$ and $k/n \to c$. The asymptotic moments of the
output matrix $Z=\Phi\otimes\overline{\Phi}(E_n)$ are given by
\begin{eqnarray*}
\trace(Z ) &=& 1;\\
\E\trace((cnZ)^2 ) &=& 2+c^2+O(n^{-1});\\
\E\trace((cnZ)^p ) &=& 1+O(n^{-1}) \qquad  \forall p \geq3.
\end{eqnarray*}
\end{theorem}

\begin{remark}
Before we prove this result, we would like to point out to readers
aware of random matrix theory and matrix integrals
that the symbol $O(n^{-1})$ is actually optimal. One can check by
inspection that there are terms of order $n^{-1}$ in the expansion
of the quantities of the theorem. This observation stresses the fact
that the matrix model $Z$ does not
behave like a usual unitarily invariant matrix model, but rather like
an orthogonal matrix model, even though the underlying invariance
group is the unitary group.\vadjust{\goodbreak} This technicality explains why we can only
obtain convergence in probability
of the rescaled largest eigenvalue and not the almost sure convergence.
\end{remark}

\begin{pf*}{Proof of Theorem~\ref{thm:asympt-Z}}
We start from the exact expression for fixed $n$ and $k$ for the
moments of $Z$ (see~\cite{collins-nechita-1}):
%
%
%e25 ###
\begin{equation}\label{eq:bi_canal_UUbar_det}
\E[\trace(Z^p)] = \sum_{\alpha, \beta\in\S_{2p}}k^{\# \alpha
}n^{\#(\alpha\gamma^{-1}) +\# (\beta\delta) - p}\Wg(\alpha\beta^{-1}).
\end{equation}
Since in this ``conjugate'' case, the permutations $\alpha$ and $\beta
$ act on the whole set of $2p$ boxes, we introduce a special labeling
on the boxes. The \textit{top} row of boxes (corresponding to the
channel $\Phi$) shall be labeled by $1^T, 2^T, \ldots, p^T$ and the
\textit{bottom} row by $1^B, 2^B, \ldots, p^B$. With this notation,
the permutations $\gamma$ and $\delta$ have the following expressions:
%
%
%e26 ###
\begin{eqnarray}
\gamma&=& \bigl(p^T \; (p-1)^T \cdots1^T\bigr) \; (1^B \; 2^B \cdots p^B);
\nonumber
\\[-8pt]
\\[-8pt]
\delta&=& (1^T \; 1^B) \; (2^T \; 2^B) \cdots(p^T \; p^B).
\nonumber
\end{eqnarray}
Dropping the number-of-cycles statistics $\#(\cdot)$ in favor of
permutation lengths $|\cdot|$, replacing $k \sim cn$ and using the
standard asymptotic expansion for the Weingarten function, we have
\[
\E[\trace(Z^p)] \sim\sum_{\alpha, \beta\in\S_{2p}} c^{-(|\alpha
|+|\alpha\beta^{-1}|)}n^{p-(|\alpha|+|\alpha\gamma^{-1}|+|\beta
\delta|+2|\alpha\beta^{-1}|)}\Mob(\alpha\beta^{-1}).
\]
In order to find the first order asymptotic (in $n$) of this
expression, one has to minimize the quantity
\[
|\alpha|+|\alpha\gamma^{-1}|+|\beta\delta|+2|\alpha\beta^{-1}|
\]
over all permutations $\alpha, \beta\in\S_{2p}$. We start by
simplifying this optimization problem over two permutations by using
the following two inequalities:
%
%e28 ###
%e27 ###
\begin{eqnarray}
\label{eq:first_geodesic_alpha}|\alpha|+|\alpha\beta^{-1}| &\geq&
|\beta|;\\
\label{eq:second_geodesic_alpha}
|\alpha\gamma^{-1}|+|\alpha\beta
^{-1}| &\geq&|\beta\gamma^{-1}|.
\end{eqnarray}
Note that these inequalities can be simultaneously saturated by
choosing, for example, $\alpha= \beta$. So, one is left with the
following minimization problem over $\beta\in\S_{2p}$:
%
%
%e29 ###
\begin{equation}\label{eq:min_beta}
\mbox{minimize}  \qquad  S_1(\beta) = |\beta|+|\beta\gamma
^{-1}|+|\beta\delta^{-1}|.
\end{equation}

The main ingredient in tackling this problem is the fact that both
permutations $\delta$ and $\gamma$ lie on the geodesic between the
identity permutation $\id$ and the full-cycle permutation
\[
\gammat= (p^T \cdots2^T 1^T 1^B 2^B \cdots p^B).
\]
This follows from the saturated triangle inequalities $|\delta| +
|\delta^{-1}\gammat| = p + p-1 = 2p-1$ and $|\gamma| + |\gamma
^{-1}\gammat| = 2(p-1) + 1 = 2p-1$.\vadjust{\goodbreak} If fact, one has $\gammat=
(p^T1^B)\cdot\gamma$. Under Biane's isomorphism,
(cf.~\ref{lem:S_p}), the
permutations $\delta$ and $\gamma$ correspond to the noncrossing
partitions in Figure~\ref{fig:delta_gamma_NC}.

%
%f9 ###
\begin{figure}

\includegraphics{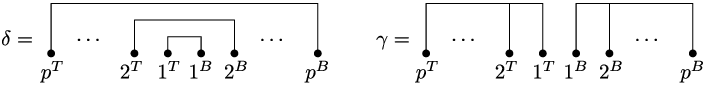}

\caption{Noncrossing partitions associated with permutations $\delta$
and $\gamma.$}
\label{fig:delta_gamma_NC}
\end{figure}

We start by looking at the following simplified minimization problem:
\[
\mbox{minimize}  \qquad S_2(\beta) = |\beta|+|\beta\delta
^{-1}|+|\beta\gammat^{-1}|.
\]
Obviously, $|\beta|+|\beta\gammat^{-1}| \geq2p-1$, with equality
iff $\beta$ lies in the geodesic between $\id$ and $\gammat$. It
follows from a parity argument that if $\beta$ is not an element of
the geodesic set $\id\to\gammat$, then $|\beta|+|\beta\gammat
^{-1}| \geq2p+1$ and hence, since in this case one has $\beta\neq
\delta$, $S_2(\beta) \geq2p+2$. If $\beta$ is a geodesic element,
then $S_2(\beta) \geq2p-1$, with equality iff $\beta=\delta$.

Since the permutations $\gamma$ and $\gammat$ are at distance one,
the same holds for $\beta\gamma^{-1}$ and $\beta\gammat^{-1}$:
\[
\beta\gamma^{-1} = \beta\gammat^{-1} \cdot(p^T1^B).
\]
We have $|\beta\gamma^{-1}| = |\beta\gammat^{-1}| \pm1$ and, even
more precisely,
\[
|\beta\gamma^{-1}| =
\cases{\displaystyle
|\beta\gammat^{-1}| - 1, & \quad  if   $p^T$   and   $1^B$
  are in the same block of   $\beta\gammat^{-1}$, \cr\displaystyle
|\beta\gammat^{-1}| + 1, & \quad  otherwise.
}
\]
Note that the condition appearing in the first case can be restated in
the following, simpler way. It is known that since $\beta$ is a
geodesic element, $\beta^{-1}\gammat$ is also on the geodesic, and
the noncrossing partition associated with $\beta^{-1}\gammat$ is the
Kreweras complement of the partition associated with $\beta$ (see
\cite{nica-speicher}, Lecture 9, for a definition of the Kreweras
complement of a noncrossing partition). The condition that $p^T$ and
$1^B$ belong to the same block of $K(\beta)$ is depicted in Figure
\ref{fig:kreweras} and it is easily seen to be equivalent to $\beta
\leq\gamma$ (the permutations are compared here via their associated
partitions).

%
%f10 ###
\begin{figure}[b]

\includegraphics{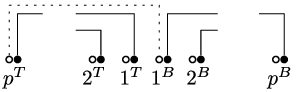}

\caption{Kreweras complement of $\beta.$}
\label{fig:kreweras}
\end{figure}

It follows that $S_1(\beta) = S_2(\beta) \pm1$ and, in order to
conclude, one needs to look at the position of the permutation $\beta$
with respect to the geodesic $\id\to\gammat$. If $\beta$ is not an
element of $\id\to\gammat$, then $S_1(\beta) \geq S_2(\beta) - 1
\geq2p+2-1=2p+1,$ which is enough to conclude. We assume from now on
that $\beta$ is a geodesic element. If $|\beta\gamma^{-1}| = |\beta
\gammat^{-1}|+1$, then $S_1(\beta) = S_2(\beta) + 1 \geq2p$, with
equality if and only if $\beta= \delta,$ and the conclusion follows.
The most difficult case is when $|\beta\gamma^{-1}| = |\beta\gammat
^{-1}|-1,$ which is equivalent to the fact that $p^T$ and $1^B$ are in
different blocks of $\beta\gammat^{-1}$. We get $S_1(\beta) = 2p-2 +
|\beta\delta^{-1}|$. We claim that for any geodesic $\beta$ such
that $\beta\leq\gamma$, $|\beta\delta^{-1}| \geq p$.
This follows from the fact that the permutation $\beta\delta
^{-1}=\beta\delta$ has no fixed points: any index $x^T$ is mapped by
$\delta$ to $x^B$, which is then mapped by $\beta$ to some $y^B \neq
x^T$ (the same holds for bottom indices). Since it has no fixed points,
each cycle of $\beta\delta^{-1}$ has cardinality at least 2, and thus
$\beta\delta^{-1}$ has at most $p$ cycles, which implies $|\beta
\delta^{-1}| = 2p - \#(\beta\delta^{-1}) \geq p$. We conclude that
if a geodesic permutation $\beta$ verifies $|\beta\gamma^{-1}| =
|\beta\gammat^{-1}|-1$, then $S_1(\beta) \geq2p-2+p \geq2p+1$ for
$p \geq3$.

So far, we have shown that the inequality $S_1(\beta) \geq2p$ holds
for all $\beta$ and $p$. Moreover, for $p\geq3$, we have shown that
equality holds if and only if $\beta= \delta$. For $p=2$, using an
exhaustive search in $\S_4$, we can identify the permutations which
saturate the equality: $\beta\in\{\id, \delta, \gamma\}$.

Now that we have completely solved the minimization problem for $\beta
$, let us go back to equations~(\ref{eq:first_geodesic_alpha}) and (\ref
{eq:second_geodesic_alpha}) and find, for each minimizing $\beta$, the
values of $\alpha$ which saturate both inequalities. For $\beta
=\delta$, the geodesic $\id\to\delta$ has a very simple expression
since $\delta$ is a product of transpositions with disjoint support
(see the proof of Theorem 6.3 in~\cite{collins-nechita-1}):
\[
\id\to\alpha\to\delta \quad \Longleftrightarrow \quad \exists
\varnothing\subseteq A \subseteq\{1, 2, \ldots, p\}  \quad \mbox{such
that} \quad  \alpha= \prod_{i \in A} (i^T i^B).
\]
Obviously, we have $\alpha\delta^{-1} = \prod_{i \notin A} (i^T
i^B)$, and thus formula~(\ref{eq:second_geodesic_alpha}) reads
$|\alpha\gamma^{-1}| + p-|A| = p$. Writing explicitly $\alpha\gamma
^{-1}$, we can show that
\[
\#(\alpha\gamma^{-1}) =
\cases{\displaystyle
1, & \quad  if   $A=\varnothing$, \cr\displaystyle
|A|, & \quad  otherwise.
}
\]
Obviously, $A=\varnothing$ does not verify the equality, so one is left
with $2p-|A|=|A| \Rightarrow|A|=p$ and hence $\alpha= \delta=\beta
$. The other two cases for $p=2$ ($\beta=\id$ and $\beta=\gamma$)
are trivial and yield the same result $\alpha=\beta$. In conclusion,
for $p \geq3$, we obtain
\[
\E[\trace(Z^p)] = c^{-p}n^{-p} + o(n^{-p})
\]
and for $p=2$,
\[
\E[\trace(Z^2)] = (1+2c^{-2})n^{-2} + o(n^{-2}),
\]
which completes the proof.
\end{pf*}

At this point, the description of the random matrix $Z$ is not
complete: the moment information of the preceding theorem allows us to
infer that there are at least some eigenvalues on the scale of $1/n$
and that the rest of the spectrum is distributed on lower scales, such
as $1/n^2$. Hayden\vadjust{\goodbreak} and Winter's proof of the existence of a \textit
{large} eigenvalue contains, as a byproduct, some information on the
eigenvector for this particular eigenvalue. Indeed, they use the
projection on the Bell state to obtain a lower bound for the largest
eigenvalue of $Z$, so one can use this projector to obtain more precise
information on the eigenvalue distribution of $Z$.

%
%f11 ###
\begin{figure}[b]

\includegraphics{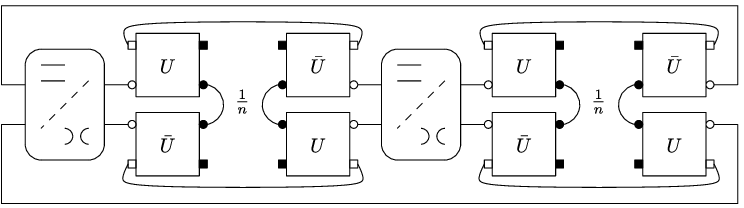}

\caption{Developing $\trace(n^2QZQ)^2$.}
\label{fig:QZQ_choice}
\end{figure}

In order to obtain information on the rest of the spectrum, we
introduce the orthogonal projection $Q = \I- E$, where $E$ is the
maximally entangled state. Using the (rank $n^2-1$) projector $Q$, we
shall obtain some information on the smallest $n^2-1$ eigenvalues of
the output matrix $Z$.

\begin{theorem}\label{thm:asmpt-QZQ}
Almost surely, the matrix $c^2n^2QZQ$ converges in distribution, to a
Marchenko--Pastur law with parameter $c^2$.
\end{theorem}

\begin{pf}
We compute the moments of the random matrix $c^2n^2QZQ$ and show that
they converge to the corresponding moments of the limit law:
\[
\lim_{n \to\iy} \frac{1}{n^2}\E\trace(c^2n^2QZQ)^p = \int x^p
\,d\pi_{c^2}(x).
\]

We start by replacing $Q=\I-E$ and expanding the product
\begin{eqnarray*}
&&\frac{1}{n^2}\E\trace(c^2n^2QZQ)^p\\
 && \qquad = c^{2p}n^{2p-2}\E\trace(\I
-E)Z(\I-E)Z\cdots(\I-E)Z \\
 && \qquad = c^{2p}n^{2p-2} \sum_{f \in\F_p}
(-1)^{|f^{-1}(E)|}n^{-|f^{-1}(E)|} \E\trace f(1) Z f(2)Z \cdots f(p) Z,
\end{eqnarray*}
where $\F$ is a set of the $2^p$ choice functions $f\dvtx \{1, 2, \ldots,
p\} \to\{\I, E\}$. Notice that in the last formula, each Bell
projector $E$ is multiplied by a factor~$-1/n$.

The moment $\E\trace f(1) Z f(2)Z \cdots f(p) Z$ is computed with our
graphical calculus, and the computation is similar to those in the
proof of Theorem~\ref{thm:asympt-UV} (see Figure~\ref{fig:QZQ_choice} for the case $p=2$):
\[
\E\trace f(1) Z f(2)Z \cdots f(p) Z = \sum_{\alpha, \beta\in\S
_{2p}}k^{\# \alpha}n^{\#(\alpha\fhat^{-1}) +\# (\beta\delta) -
p}\Wg(\alpha\beta^{-1}),\vadjust{\goodbreak}
\]
where $\fhat\in\S_{2p}$ is the permutation associated with the
choice function $f \in\F_p,$ describing the way $f$ connects the
different instances of the channel. The exact action of $\fhat$ can be
easily computed:
\begin{eqnarray*}
i^T &\stackrel{\fhat}{\mapsto}&
\cases{\displaystyle
(i-1)^T, & \quad   if   $f(i)=\I$,\cr\displaystyle
i^B, & \quad   if   $f(i)=E$,
}
\\
i^B &\stackrel{\fhat}{\mapsto}&
\cases{\displaystyle
(i+1)^B, & \quad   if   $f(i+1)=\I$,\cr\displaystyle
i^T, & \quad   if   $f(i+1)=E$,
}
\end{eqnarray*}
where the arithmetic operations of indices $i$ should be understood
modulo~$p$.

When trying to compute the leading order terms in the expression of  $\E
\trace(n^2QZQ)^p$, we have to understand the possible cancellations of
high powers in $n$. When writing the exact formula for the $p$th moment
and separating the $(\alpha, \beta)$ and $f$ parts, we get
\begin{eqnarray*}
 \frac{1}{n^2}\E\trace(c^2n^2QZQ)^p &=& c^{2p}\sum_{\alpha, \beta\in
\S_{2p}} n^{5p-2-|\beta\delta|}k^{2p-|\alpha|} \Wg(\alpha\beta
^{-1})\\
&&\hphantom{c^{2p}\sum_{\alpha, \beta\in
\S_{2p}}}
{}\times \sum_{f \in\F_p} (-1)^{|f^{-1}(E)|} n^{-(|f^{-1}(E)| +
|\alpha\fhat^{-1}|)}.
\end{eqnarray*}

Note that the sum over $f \in\F_p$ depends only on the permutation
$\alpha$. Next, we show that for a large class of permutations $\alpha
$ [the ones which are responsible for the \textit{large eigenvalue} of
size $1/(cn)$ of Theorem~\ref{thm:asympt-Z}], this sum is zero.
Let us introduce the set of ``vertical line permutations,''
\begin{eqnarray*}
\V&=& \{\sigma\in\S_{2p} \mid\exists i \in\{1, \ldots, p\}
\mbox{ such that } \sigma(i^T) = i^B \mbox{ or } \sigma(i^B) = i^T
\} \\
&=& \{\sigma\in\S_{2p} \mid\sigma\delta\mbox{ has at least one
fixed point}\}.
\end{eqnarray*}
Fix a permutation $\alpha\in\V$ and some index $i$ such that $\alpha
(i^T) = i^B$ or $\alpha(i^B) = i^T$. Consider the ``flip at position
$i$'' involution $T_i\dvtx \F_p \to\F_p$ which maps a choice function $f$
to the function
\[
T_i f\dvtx j \mapsto
\cases{\displaystyle
\I, & \quad   if   $j=i$   and   $f(j) = E$,\cr\displaystyle
E, & \quad   if   $j=i$   and   $f(j) = \I$,\cr\displaystyle
f(j), & \quad   if   $j \neq i$.
}
\]
We shall show that
\begin{eqnarray*}
&&\sum_{f \in\F_p} (-1)^{|f^{-1}(E)|} n^{-(|f^{-1}(E)| + |\alpha\fhat
^{-1}|)}\\
&& \qquad  = \sum_{f \in\F_p} (-1)^{|(T_i f)^{-1}(E)|} n^{-(|(T_i
f)^{-1}(E)| + |\alpha\widehat{(T_i f)}^{-1}|)},
\end{eqnarray*}
which will imply that for all $\alpha\in\V$, both sums are zero.
Since the cardinalities of the sets $f^{-1}(E)$ and $(T_i f)^{-1}(E)$
differ by\vadjust{\goodbreak} exactly one, all we need to show is that for all $f \in\F_p$,
\[
|f^{-1}(E)| + |\alpha\fhat^{-1}|= |(T_i f)^{-1}(E)| + |\alpha
\widehat{(T_i f)}^{-1}|.
\]
To this end, notice that (the order in which one multiplies the
transpositions in not important)
\[
\fhat= \prod_{j \dvtx  f(j)=E} \bigl( (j-1)^T \; j^B \bigr)
\gammat
\]
and hence $\alpha\widehat{(T_i f)}^{-1} = \alpha\fhat^{-1} \cdot
( (i-1)^T \; i^B )$. From this we find that
\[
|\alpha\widehat{(T_i f)}^{-1}| =
\cases{\displaystyle
|\alpha\fhat^{-1}| - 1, & \quad   if   $(i-1)^T$   and   $i^B$
  belong to\cr & \quad  the same orbit of   $\alpha\fhat^{-1}$,\cr\displaystyle
|\alpha\fhat^{-1}| + 1, & \quad   otherwise.
}
\]
Let us now suppose that $\alpha(i^T) = i^B$, the other case $\alpha
(i^B) = i^T$ being similar. If $f(i)=\I$, then $\fhat(i^T) =
(i-1)^T$, $\alpha\fhat^{-1}((i-1)^T )= i^B$ and thus $|\alpha\fhat
^{-1}| - |\alpha\widehat{(T_i f)}^{-1}| = 1$. On the other hand,
$f(i)=\I\Rightarrow(T_i f)(i) = E$ and then $|f^{-1}(E)| - |(T_i
f)^{-1}(E)| = -1,$ and we see that the differences compensate. The case
$f(i) = E$ is treated in a similar manner.

We have proven that for all permutations $\alpha\in\V$, the sum over
all choices $f \in\F_p$ is exactly zero; notice that the computations
we have carried out up to this point are \textit{nonasymptotic}; they
are true at fixed matrix sizes $n$ and~$k$. We now interchange the sums
over $(\alpha, \beta)$ and $f$, replace $k=cn$ and use the first-order asymptotic for the Weingarten function:
\begin{eqnarray*}
&&\frac{1}{n^2}\E\trace(c^2n^2QZQ)^p \\
&& \qquad \sim \sum
_{\alpha, \beta\in\S_{2p}, \alpha\notin\V}
n^{3p-2-(|\beta\delta| + |\alpha| + 2|\alpha\beta^{-1}|)}
c^{2p-(|\alpha| + |\alpha\beta^{-1}|)} \Mob(\alpha\beta^{-1})
\\
&& \qquad \quad  \hphantom{\sum
_{\alpha, \beta\in\S_{2p}, \alpha\notin\V}}{}\times \sum_{f \in\F_p} (-1)^{|f^{-1}(E)|} n^{-(|f^{-1}(E)| + |\alpha
\fhat^{-1}|)}.
\end{eqnarray*}

To obtain the dominant power of $n$, we must minimize the following
quantity over $(\alpha, \beta, f) \in(\S_{2p} \setminus\V) \times
\S_{2p} \times\F_p$:
\[
S(\alpha, \beta, f) = |\beta\delta| + |\alpha| + 2|\alpha\beta
^{-1}| + |f^{-1}(E)| + |\alpha\fhat^{-1}|.
\]
Since $\alpha\notin\V$, $\alpha\delta$ has no fixed point and
hence $|\alpha\delta| \geq p$. Using the facts that $|\alpha\beta
^{-1}| + |\beta\delta| \geq|\alpha\delta|$, $|\alpha\beta^{-1}|
\geq0$ and $|\alpha|+|\alpha\fhat^{-1}| \geq|\fhat|$, we obtain that
\[
S(\alpha, \beta, f) \geq p + |f^{-1}(E)| + |\fhat|
\]
with equality if and only if $\beta= \alpha$, $|\alpha\delta| = p$
and $\alpha$ is on the geodesic between $\id$ and $\fhat$. On the
other hand, we can easily compute the number\vadjust{\goodbreak} of cycles of $\fhat$:
\[
\# \fhat=
\cases{\displaystyle
2, & \quad   if   $f \equiv\I$,\cr\displaystyle
|f^{-1}(E)|, & \quad   otherwise .
}
\]
Hence, $S(\alpha, \beta, f) \geq3p-2$, with equality if and only if
$f \equiv\I$, $\beta= \alpha$, $|\alpha\delta| = p$ and $\alpha$
is a permutation on the geodesic $\id\to\hat{\I} = \gamma$. Since
$\gamma= \gamma^T \oplus\gamma^B$ is a disjoint union of the two
cycles $\gamma^T = (p^T \cdots2^T 1^T)$ and $\gamma^B = (1^B 2^B
\cdots p^B) = (\gamma^T)^{-1}$, the condition that $\alpha$ should be
a geodesic permutation amounts to $\alpha= \alpha^T \oplus\alpha
^B$, where $\alpha^{T,B} \in\S_p$ are geodesic permutations with
respect to the cycles $\gamma^{T,B}$. We can easily show that $\#
((\alpha^T \oplus\alpha^B) \delta) = \#(\alpha^T \alpha^B)$,
where the first permutation is an element of $\S_{2p}$ and the second
one is an element of $\S_p$. Using this equality, the condition
$|\alpha\delta| = p$ implies that $\alpha^T \alpha^B = \id_p$.
When putting all these considerations together, one obtains the final
formula for the dominant term of the $p$th moment of $QZQ$:
\[
\frac{1}{n^2}\E\trace(c^2n^2QZQ)^p \sim\sum_{\id\to\alpha^T \to
\gamma^T} c^{2p-2 |\alpha^T|} \Mob(\id) = \sum_{\id\to\alpha^T
\to\gamma^T} c^{2\#\alpha^T}.
\]
Following the proof of Theorem~\ref{thm:asympt-UV}, the moments of the
Marchenko--Pastur distribution of parameter $c^2$ are easily recognized
and the convergence in moments is settled. The proof of the almost sure
convergence is more involved and can be found in the \hyperref[appm]{Appendix}.
\end{pf}

From this we deduce the following theorem, which summarizes all the
results obtained thus far in this section.

\begin{theorem}
\label{thm:asymptotics-Z}
The eigenvalues $\lambda_1\geq\cdots\geq\lambda_{n^2}$ of $Z$ are
such that:
\begin{itemize}
\item
in probability, $cn\lambda_1\to1;$
\item
almost surely,
$\frac{1}{n^2-1}\sum_{i=2}^{n^2}\delta_{c^2n^2\lambda_i}$ converges
to a free Poisson distribution with parameter $c^2$.
\end{itemize}
\end{theorem}

\begin{pf}
Let $\tilde{\lambda}_1\geq\cdots\geq\tilde{\lambda}_{n^2-1}$ be
the eigenvalues of $QZQ$, seen as a matrix in $\M_{n^2-1}(\C)$. By
Cauchy's interlacing theorem (\cite{bhatia},
Corollary III.1.5), the eigenvalues of $QZQ$ and those of $Z$ are
intertwined and satisfy
\[
\lambda_1 \geq\tilde\lambda_1 \geq\lambda_2 \geq\cdots\geq
\lambda_{n^2-1} \geq\tilde\lambda_{n^2-1} \geq\lambda_{n^2}.
\]
Therefore, the second statement follows immediately from Theorem \ref
{thm:asmpt-QZQ}.

For the first statement, we have
\[
1\leq cn\lambda_1\leq c^3n^3\lambda_1^3\leq c^3n^3\lambda_1^3+\cdots+
c^3n^3\lambda_{n^2}^3,
\]
so the inequality obtains if one takes expectations. In addition, we
know from Theorem~\ref{thm:asympt-Z} that
$\E[c^3n^3Z^3]=1+O(n^{-1})$, and
therefore
\[
\E[cn\lambda_1]=1+O(n^{-1}).
\]
This proves the first statement.\vadjust{\goodbreak}
\end{pf}

An important result for quantum information theoretic purposes is the following.
\begin{proposition}
\label{prop:entropy_bicanal}
Almost surely, in the limit $n \to\iy$, the von Neumann entropy of
the matrix $Z$ satisfies
\[
H(Z) =
\cases{\displaystyle
2\log n - \frac{1}{2c^2}+o(1), & \quad   if   $c\geq1$,\cr\displaystyle
2\log(cn) - \frac{c^2}{2}+o(1), & \quad   if   $0<c<1$.
}
\]
\end{proposition}

\begin{pf}
We use the fact that $cn\lambda_1\geq1$.
Since $x\log x\leq x^3-1$ for any $x\geq1$, we have
\[
cn\lambda_1\log(cn\lambda_1)\leq(cn\lambda_1)^3-1\leq\sum
_{i=1}^{n^2} (cn\lambda_i)^3-1.
\]
Taking the expectation and using Theorem~\ref{thm:asympt-Z}, we get
\[
\E[cn\lambda_1\log(cn\lambda_1)]=O(n^{-1}).
\]
Similarly, we know by Theorem~\ref{thm:asympt-Z} that
\[
\E[\lambda_1]=O(n^{-1}).
\]
Putting this together, we obtain
\[
\E[-\lambda_1\log\lambda_1]=o(1).
\]

We are now left with evaluating
\[
\E[-\lambda_2\log\lambda_2 - \cdots- \lambda_{n^2}\log\lambda_{n^2}].
\]
This can be done exactly in the same way as in Corollary \ref
{cor-entropy}, and we then obtain the desired formula.
\end{pf}

\begin{remark}
It is important to remark here that the estimate of Propositions
\ref{prop:entropy_bicanal} and
\ref{prop:entropy_independent}
are asymptotically the same. This implies that in this scaling,
the choice $U_1=\overline{U_2}$ is irrelevant in the construction of
counterexamples to the additivity problem.
However, it remains to be checked whether this scaling indeed yields
counterexamples with high probability, and
this is not clear from our first-order asymptotics.
\end{remark}

\begin{appendix}
\section*{Appendix}\label{appm}

In this appendix we present the complete proofs of the almost sure
convergence statements in Theorems~\ref{thm:asympt-UV} and \ref
{thm:asmpt-QZQ}.

\begin{pf*}{Proof of Theorem~\ref{thm:asympt-UV}, continued (Almost
sure convergence)}
We have already proven the convergence \textit{in moments}.
To prove the almost sure convergence,\vadjust{\goodbreak} it is sufficient to show that for
each $p,$ the series of covariance of the $p$ moments is convergent. A~classical application of the Borel--Cantelli lemma then suffices to
complete the proof.

We start with the simplest term, $\E[\tracenorm
_{n^2}((c^2n^2Z)^p)]^2$. Since we need to compute its first two terms in
the asymptotic expansion in $n$, we look at the subleading term
($n^{-1}$) of $\E[\tracenorm_{n^2}((c^2n^2Z)^p)]$. Such terms come
from permutations for which the exponent $\mathcal P_n$ has value
$2(p-1)+1 = 2p-1$. Analyzing equations~(\ref{eq:UV_ineq_first})--(\ref
{eq:UV_ineq_last}), we see that the permutations which ``almost''
saturate the bound are those which verify $\id\to\alpha_U=\beta_U
\to\gamma$, $\id\to\alpha_V=\beta_V \to\gamma$ and $|\beta
_U^{-1}\beta_V|=|\alpha_U^{-1}\alpha_V|=1$. In conclusion, we have
\begin{eqnarray*}
&&\E[\tracenorm_{n^2}((c^2n^2Z)^p)]\\
&& \qquad  \sim\sum_{\id\to\alpha_U=\beta
_U=\alpha_V=\beta_V \to\gamma} c^{2\#\alpha_U} + n^{-1} \mathop{\mathop{\mathop{\sum}_{\id\to\alpha_U=\beta_U \to
\gamma}}_{
\id\to\alpha_V=\beta_V \to\gamma}}_{ |\beta_U^{-1}\beta_V|=1}
c^{\#\alpha_U+\#\alpha_V} + O(n^{-2}).
\end{eqnarray*}
Taking the square gives
%
%e30 ###
\begin{eqnarray}
\label{eq:UV_sq}
&&\E[\tracenorm_{n^2}((c^2n^2Z)^p)]^2 \nonumber\\
&& \qquad = \biggl[\sum
_{\id\to\alpha_U=\beta_U=\alpha_V=\beta_V \to\gamma} c^{2\#
\alpha_U} \biggr]^2\nonumber
\\[-8pt]
\\[-8pt]
&& \qquad  \quad {}+ n^{-1} 2 \biggl[\sum_{\id\to\alpha_U=\beta_U=\alpha
_V=\beta_V \to\gamma} c^{2\#\alpha_U} \biggr]\cdot\biggl[\mathop{\mathop
{\mathop{\sum}_{\id\to\alpha_U=\beta_U \to\gamma}} _{\id\to
\alpha
_V=\beta_V \to\gamma}}_{ |\beta_U^{-1}\beta_V|=1} c^{\#\alpha_U+\#
\alpha_V} \biggr]\nonumber\\
&& \qquad  \quad {} + O(n^{-2}).
\nonumber
\end{eqnarray}
In order to compute the asymptotic expansion of $\E[
(\tracenorm_{n^2}((c^2n^2Z)^p )^2 )]$, we must consider two
copies of the diagram corresponding to the $p$th power of $c^2n^2Z$.
The boxes are originally connected by the permutation
\[
\gamma_2 = (p \; p-1 \; \cdots\; 2 \; 1) (2p \; 2p-1 \; \cdots\; p+2
\; p+1)\in\S_{2p}.
\]
After counting the loops, we finds an analogous formula for the mean trace,
\begin{eqnarray*}
&&\E[ (\tracenorm_{n^2}((c^2n^2Z)^p )^2) ]\\
&& \qquad  \sim
\sum_{\alpha_U, \beta_U, \alpha_V, \beta_V \in\S_{2p}}
n^{4p-4-\mathcal P_{n,2}} c^{4p-\mathcal P_{c,2}} \Mob(\alpha_U\beta
_U^{-1})\Mob(\alpha_V\beta_V^{-1}),
\end{eqnarray*}
where
\[
\mathcal P_{n,2} = |\alpha_U| + |\alpha_V| + |\gamma_2^{-1}\alpha
_U| + |\gamma_2^{-1}\alpha_V| + |\beta_U^{-1}\beta_V| + 2|\alpha
_U\beta_U^{-1}| + 2|\alpha_V\beta_V^{-1}|
\]
and
\[
\mathcal P_{c,2} = |\alpha_U| + |\alpha_V| + |\alpha_U\beta_U^{-1}|
+ |\alpha_V\beta_V^{-1}|.\vadjust{\goodbreak}
\]
Using the same inequalities and arguments as above, we find that
$\mathcal P_{n,2} \geq2|\gamma_2| = 4p-4$ with equality iff $\id
\alpha_U = \beta_U = \alpha_V = \beta_V \to\gamma_2$ is a
geodesic. Since $\gamma_2$ contains two $p$-cycles, the preceding
condition is equivalent to $\alpha_U = \beta_U = \alpha_V = \beta_V
= \alpha\oplus\alpha'$, where $\alpha\in\S_p$ and $\alpha' \in
\S\{p+1, p+2, \ldots, 2p\} \isom\S_p$ are such that $\id\to\alpha
\to\gamma$ and $\id\to\alpha' \to\gamma$ are geodesics. Since M\"obius
functions vanish, this dominating term is equal to the first term
in the asymptotic expansion~(\ref{eq:UV_sq}) of $\E[\tracenorm
_{n^2}((c^2n^2Z)^p)]^2$. The term responsible for the $n^{-1}$
contribution comes from permutations $\alpha_U, \beta_U, \alpha_V,
\beta_V \in\S_{2p}$ such that $\id\to\alpha_U=\beta_U \to\gamma
_2$, $\id\to\alpha_V=\beta_V \to\gamma_2$ and $|\beta
_U^{-1}\beta_V|=|\alpha_U^{-1}\alpha_V|=1$. Since, from the geodesic
condition, $\alpha_U = \alpha'_U \oplus\alpha''_U$ and $\alpha_V =
\alpha'_V \oplus\alpha''_V$, the condition $|\alpha_U^{-1}\alpha
_V|=1$ is equivalent to either
\[
\alpha'_U = \alpha'_V  \quad \mbox{and} \quad  |(\alpha
''_U)^{-1}\alpha''_V|=1
\]
or
\[
|(\alpha'_U)^{-1}\alpha'_V|=1  \quad \mbox{and} \quad  \alpha''_U =
\alpha''_V.
\]
Summing these contributions, we find the term in $n^{-1}$ from equation
(\ref{eq:UV_sq}). Hence, the dominating ($n^0$) and the subdominating
($n^{-1}$) terms from\break $\E[\tracenorm_{n^2}((c^2n^2Z)^p)]^2$ and $\E
[ (\tracenorm_{n^2}((c^2n^2Z)^p )^2 )]$ are
equal, which implies that the general term of the series of covariances
has order $n^{-2}$. The series is thus summable and a
Borel--Cantelli-type argument completes the proof of the almost sure
convergence from Theorem~\ref{thm:asympt-UV}.
\end{pf*}

\begin{pf*}{Proof of Theorem~\ref{thm:asmpt-QZQ}, continued (Almost
sure convergence)}
We now prove the almost sure convergence statement of Theorem \ref
{thm:asmpt-QZQ}. We use the same technique as before, showing that the
covariance series converges. The first step is to analyze the
subleading terms ($n^{-1}$) in the expression of the $p$th moment for
one copy of the channel. Recall that the exponent of $n$ was given by
the expression
\[
S(\alpha, \beta, f) = |\beta\delta| + |\alpha| + 2|\alpha\beta
^{-1}| + |f^{-1}(E)| + |\alpha\fhat^{-1}|.
\]
Using the triangle inequality $|\alpha|+ |\alpha\fhat^{-1}| \geq
|\fhat|$, we split this minimization task into two independent problems:
\[
\mbox{minimize}  \qquad |\beta\delta| + 2|\alpha\beta^{-1}|
\]
and
\[
\mbox{minimize}  \qquad |f^{-1}(E)| + |\fhat|.
\]
The $2p-2$ minimum in the second problem is reached for $f \equiv\I$;
if $f$ is different from $\I$, it follows from the above analysis that
$|f^{-1}(E)| + |\fhat| \geq2p$ and thus only $f \equiv\I$
contributes to the subleading $n^{-1}$ term. Moreover, a~parity
argument for the geodesic inequality $|\alpha|+ |\alpha\fhat^{-1}|
\geq|\fhat|$ implies that the permutation $\alpha$ must lie on the
geodesic between $\id$ and $\hat{\I} = \gamma$. Let us now describe
the couples $(\alpha, \beta) \in\S_{2p}^2$ such that $|\beta\delta
| + 2|\alpha\beta^{-1}| = p+1$. Since $|\beta\delta| + |\alpha
\beta^{-1}| \geq|\alpha\delta| \geq p$, we need to consider two cases.\vadjust{\goodbreak}

In the first case, we assume that $|\alpha\delta| = p+1$ and $\alpha
=\beta$. Since $\alpha$ is a geodesic permutation, $\alpha= \alpha
^T \oplus\alpha^B$ and the condition $|\alpha\delta| = p+1$ is
equivalent to $|\alpha^T \alpha^B|=1$. In conclusion, this case gives
a contribution of
\[
\frac{1}{n} \mathop{\mathop{\mathop{\sum}_{\id\to\alpha^T \to
\gamma^T }}_{ \id
\to\alpha^B \to\gamma^B}}_{ |\alpha^T \alpha^B|=1} c^{\#\alpha^T
+ \#\alpha^B}.
\]

In the second case, $|\alpha\delta| = p$ and $|\alpha\beta^{-1}| =
1$. This corresponds to $\alpha^T = (\alpha^B)^{-1}$, $|\alpha\beta
^{-1}| = 1$ and $|\beta\delta| = p$. Since $\beta$ is at distance 1
from $\alpha$, $\beta= \alpha(i^s \; j^t)$ for some $i, j \in\{1,
\ldots, p\}$ and $s,t \in{T,B}$. If $s=t$, then $\beta\notin\V$
and thus $|\beta\delta|\geq p,$ which is impossible. We can now
assume that $\beta= \alpha(i^T \; j^B)$ for some $i,j$. In order to
have $|\beta\delta| < p$, the permutation $\beta\delta$ must have
at least two fixed points. Using
\[
[\alpha^T \oplus(\alpha^T)^{-1} \cdot(i^T \; j^B)
](k^T) =
\cases{\displaystyle
(\alpha^T(k))^T, & \quad   if   $k \neq i$, \cr\displaystyle
((\alpha^T)^{-1}(j))^B, & \quad   if   $k = i$
}
\]
and
\[
[\alpha^T \oplus(\alpha^T)^{-1} \cdot(i^T \; j^B)
](k^B) =
\cases{\displaystyle
((\alpha^T)^{-1}(k))^B, &\quad  if   $k \neq j$, \cr\displaystyle
(\alpha^T(i))^T, &\quad  if  $ k = j$,
}
\]
we conclude that in order to get an $n^{-1}$ contribution, we must have
$\alpha^T(i) = j$. Hence, for each geodesic permutation $\alpha,$ we
can find $p$ permutations $\beta$ such that $|\alpha\beta^{-1}| = 1$
and $|\beta\delta| = p-1$. We obtain a total contribution of
[use $\Mob(\mathrm{transposition}) = -1$]
\[
-\frac{p}{n} \sum_{\id\to\alpha^T \to\gamma^T} c^{2\#\alpha^T-1}.
\]

Putting the first- and second-order contributions together, we obtain
the asymptotic expansion for the square of the expected normalized trace:
%
%e31 ###
\begin{eqnarray}\label{eq:as_conv_QZQ}
\hspace{10pt}\E[\tracenorm_{n^2}(c^2n^2QZQ)^p]^2\nonumber
 &=& \biggl[ \sum_{\id\to\alpha
^T \to\gamma^T} c^{2\#\alpha^T} \biggr]^2\\
  &&{}+\frac{2}{n} \biggl[ \sum_{\id\to\alpha^T \to\gamma^T} c^{2\#
\alpha^T} \biggr]\\
&& \quad {}\times\biggl[ \mathop{\mathop{\mathop{\sum}_{\id\to\alpha
^T \to
\gamma^T }}_{ \id\to\alpha^B \to\gamma^B}}_{|\alpha^T \alpha
^B|=1} c^{\#\alpha^T + \#\alpha^B} - p \sum_{\id\to\alpha^T \to
\gamma^T} c^{2\#\alpha^T -1} \biggr].\nonumber
\end{eqnarray}

Let us now analyze the second term in the expression of the covariance,
$\E[(\tracenorm_{n^2}(c^2n^2QZQ)^p)^2]$. The exponent we want to
minimize in this situation~is
\[
S^{(2)}(\alpha, \beta, f) = \bigl|\beta\delta^{(2)}\bigr| + |\alpha| +
2|\alpha\beta^{-1}| + |f^{-1}(E)| + |\alpha\fhat^{-1}|,
\]
where $\alpha, \beta$ are permutations in $\S_{4p}$, and the choice
function $f\dvtx \{1, \ldots, p, p+1, \ldots, 2p\} \to\{\I, E\}$ encodes
the way the $Z$ boxes are connected. Note, however, that in this case,
the diagram under consideration has at least two connected components
since we are dealing with a product of traces. Considerations similar
to the ones in the proof of the convergence in moments lead to the
conclusion that permutations $\alpha\in\V^{(2)}$ do not contribute,
so we can restrain our minimization problem to the set $(\S_{4p}
\setminus\V^{(2)}) \times\S_{4p} \times\F_{2p}$. Using the
triangle inequality $|\alpha|+ |\alpha\fhat^{-1}| \geq|\fhat|$, we
again split our problem into two independent parts: one minimization
problem for the choice function $f$ and another for the couple $(\alpha
, \beta)$. The minimization problem for $f$ is the same as in the
single channel case, with the difference that $f$ is now defined on a
set of cardinality $2p$. The quantity $|f^{-1}(E)| + |\fhat|$ is
minimized for $f \equiv\I$ and the minimum is equal to $4p-4$. Notice
that in this case, the corresponding permutation $\hat{\I}$ has the
following cycle structure: $\hat{\I} = \gamma^{T, 1} \oplus\gamma
^{T, 2} \oplus\gamma^{B, 1} \oplus\gamma^{B, 2}$, where
\begin{eqnarray*}
\gamma^{T, 1} &=& \bigl( p^T \; (p-1)^T \; \cdots\; 1^T \bigr); \\
\gamma^{T, 2} &=& \bigl( (2p)^T \; (2p-1)^T \; \cdots\; (p+1)^T
\bigr); \\
\gamma^{B, 1} &=& ( 1^B \; 2^B \; \cdots\; p^B ); \\
\gamma^{B, 2} &=& \bigl( (p+1)^B \; (p+2)^B \; \cdots\; (2p)^B \bigr).
\end{eqnarray*}
Geodesic permutations $\id_{4p} \to\alpha\to\hat{\I}$ share the
same cyclic decompositions and we can easily find the dominating term
in this case:
\[
\E[(\tracenorm_{n^2}(c^2n^2QZQ)^p)^2] = \mathop{\mathop{\sum}_{\id_{p}
\to\alpha^{T, 1} \to\gamma^{T, 1} \isom\gamma^T }}_{ \id\to
\alpha
^{T,2} \to\gamma^{T,2}\isom\gamma^T} c^{\#\alpha^{T,1} + \#\alpha
^{T,2}} + o(1),
\]
which is the same as the first term in equation (\ref
{eq:as_conv_QZQ}). Let us now move on to the subleading term in the
asymptotic expansion of $\E[(\tracenorm_{n^2}(c^2n^2QZQ)^p)^2]$. As
in the previous case, $f \equiv\I$ and contributing couples $(\alpha
, \beta)$ are of two types:\vspace*{1pt} permutations such that $|\alpha\delta
^{(2)}| = 2p+1$ and $\alpha= \beta,$ or couples such that $|\beta
\delta^{(2)}| = 2p-1$ and $|\alpha\beta^{-1}| = 1$.

The analysis of the first situation is simpler: the cycle structure of
the geodesic permutation $\alpha$ implies that $|\alpha\delta^{(2)}|
= 2p + |\alpha^{T,1}\alpha^{B,1}| + |\alpha^{T,2}\alpha^{B,2}|$.
Hence, only one of $|\alpha^{T,1}\alpha^{B,1}|$ or $|\alpha
^{T,2}\alpha^{B,2}|$ is equal to 1, the other being 0. This
corresponds to a contribution of (we use the symmetry $1
\leftrightarrow2$ of the problem)
\[
\frac{2}{n} \biggl[ \sum_{\id\to\alpha^{T,1} \to\gamma^{T,1}}
c^{2\#\alpha^{T,1}} \biggr]\cdot\biggl[ \mathop{\mathop{\mathop{\sum}_{\id
\to
\alpha^{T,2} \to\gamma^{T,2} }}_{ \id\to\alpha^{B,2} \to\gamma
^{B,2}}}_{ |\alpha^{T,2} \alpha^{B,2}|=1} c^{\#\alpha^{T,2} + \#
\alpha^{B,2}} \biggr].
\]

The second contribution is calculated in a similar manner to the case
of a single trace. Permutations $\beta$ at distance 1 from geodesic
$\alpha= \alpha^{T, 1} \oplus\alpha^{T, 2} \oplus\alpha^{B, 1}
\oplus\alpha^{B, 2}$ such that $\alpha^{B, 1} = (\alpha^{T,
1})^{-1}$ and $\alpha^{B, 2} = (\alpha^{T, 2})^{-1}$ are of the form
$\beta= \alpha(i^s \; j^t)$. The condition $|\beta\delta^{(2)}| =
2p-1$ implies that we can choose the transposition $(i^T \; j^B)$ and
that $[\alpha^{T, 1} \oplus\alpha^{T, 2}](i) = j$. This last
condition implies that $i$ and $j$ have to be in the same half of the
set $\{1, \ldots, p, p+1, \ldots, 2p\}$ and, again using the symmetry
between the first and the second trace, we can write the final
contribution:
\[
-\frac{p}{n} \biggl[ \sum_{\id\to\alpha^{T,1} \to\gamma^{T,1}}
c^{2\#\alpha^{T,1}} \biggr]\cdot\biggl[ \sum_{\id\to\alpha^{T,2}
\to\gamma^{T,2}} c^{2\#\alpha^{T,2} -1} \biggr].
\]
Summing the leading ($n^0$) and the subleading ($n^{-1}$) contributions
and comparing to equation~(\ref{eq:as_conv_QZQ}), we find that
\[
\E[(\tracenorm_{n^2}(c^2n^2QZQ)^p)^2] - \E[\tracenorm
_{n^2}(c^2n^2QZQ)^p]^2 = O(n^{-2}).
\]
The convergence of the covariance series follows, completing the proof.
\end{pf*}
\end{appendix}

\section*{Acknowledgments}

The authors would like to thank the organizers of the workshop
\textit{Thematic Program on Mathematics in Quantum Information} at the
Fields Institute, where some of this work was done.

% imsref loaded by smiklovaite, 2010-11-08 17:57:47
%

\printaddresses


\begin{thebibliography}{00}

%b1 ###
\bibitem{bhatia}
%
\begin{bbook}[mr]
\bauthor{\bsnm{Bhatia},~\bfnm{Rajendra}\binits{R.}}
(\byear{1997}).
\btitle{Matrix Analysis}.
\bseries{Graduate Texts in Mathematics}
\bvolume{169}.
\bpublisher{Springer}, \baddress{New York}.
\bid{mr={1477662}}
\end{bbook}
%
\endbibitem

%b2 ###
\bibitem{zycbook}
%
\begin{bbook}[mr]
\bauthor{\bsnm{Bengtsson},~\bfnm{Ingemar}\binits{I.}} \AND\bauthor
{\bsnm{{\.Z}yczkowski},~\bfnm{Karol}\binits{K.}}
(\byear{2006}).
\btitle{Geometry of Quantum States}.
\bpublisher{Cambridge Univ. Press}, \baddress{Cambridge}.
\bid{doi={10.1017/CBO9780511535048}, mr={2230995}}
\end{bbook}
%
\endbibitem

%b3 ###
\bibitem{bozejko}
%
\begin{barticle}[mr]
\bauthor{\bsnm{Bo{\.z}ejko},~\bfnm{Marek}\binits{M.}},
\bauthor{\bsnm{Krystek},~\bfnm{Anna~Dorota}\binits{A.~D.}} \AND
\bauthor{\bsnm{Wojakowski},~\bfnm{{\L}ukasz~Jan}\binits{{\L}.~J.}}
(\byear{2006}).
\btitle{Remarks on the {$r$} and {$\Delta$} convolutions}.
\bjournal{Math. Z.}
\bvolume{253}
\bpages{177--196}.
\bid{doi={10.1007/s00209-005-0898-2}, mr={2206642}}
\end{barticle}
%
\endbibitem

%b4 ###
\bibitem{brandao-horodecki}
%
\begin{barticle}[mr]
\bauthor{\bsnm{Brand{\~a}o},~\bfnm{Fernando G. S.~L.}\binits{F.~G.
S.~L.}} \AND
\bauthor{\bsnm{Horodecki},~\bfnm{Micha{\l}}\binits{M.}}
(\byear{2010}).
\btitle{On {H}astings' counterexamples to the minimum output entropy additivity
conjecture}.
\bjournal{Open Syst. Inf. Dyn.}
\bvolume{17}
\bpages{31--52}.
\bid{doi={10.1142/S1230161210000047}, mr={2654951}}
\bptnote{check year}
\end{barticle}
%
\endbibitem

%b5 ###
\bibitem{braunstein}
%
\begin{barticle}[mr]
\bauthor{\bsnm{Braunstein},~\bfnm{Samuel~L.}\binits{S.~L.}}
(\byear{1996}).
\btitle{Geometry of quantum inference}.
\bjournal{Phys. Lett. A}
\bvolume{219}
\bpages{169--174}.
\bid{doi={10.1016/0375-9601(96)00365-9}, mr={1408400}}
\end{barticle}
%
\endbibitem

%b6 ###
\bibitem{bryc}
%
\begin{barticle}[mr]
\bauthor{\bsnm{Bryc},~\bfnm{W{\l}odzimierz}\binits{W.}}
(\byear{2008}).
\btitle{Asymptotic normality for traces of polynomials in independent complex
{W}ishart matrices}.
\bjournal{Probab. Theory Related Fields}
\bvolume{140}
\bpages{383--405}.
\bid{doi={10.1007/s00440-007-0068-z}, mr={2365479}}
\end{barticle}
%
\endbibitem

%b7 ###
\bibitem{coecke}
%
\begin{bincollection}[mr]
\bauthor{\bsnm{Coecke},~\bfnm{Bob}\binits{B.}}
(\byear{2006}).
\btitle{Kindergarten quantum mechanics---lecture notes}.
In \bbooktitle{Quantum Theory: Reconsideration of Foundations---3}.
\bseries{AIP Conf. Proc.}
\bvolume{810}
\bpages{81--98}.
\bpublisher{Amer. Inst. Phys.}, \baddress{Melville, NY}.
\bid{mr={2212821}}
\end{bincollection}
%
\endbibitem

%b8 ###
\bibitem{collins-imrn}
%
\begin{barticle}[mr]
\bauthor{\bsnm{Collins},~\bfnm{Beno{\^{\i}}t}\binits{B.}}
(\byear{2003}).
\btitle{Moments and cumulants of polynomial random variables on
unitary groups,
the {I}tzykson--{Z}uber integral, and free probability}.
\bjournal{Int. Math. Res. Not.}
\bvolume{17}
\bpages{953--982}.
\bid{doi={10.1155/S107379280320917X}, mr={1959915}}
\end{barticle}
%
\endbibitem

%b9 ###
\bibitem{collins-nechita-1}
%
\begin{barticle}[mr]
\bauthor{\bsnm{Collins},~\bfnm{Beno{\^{\i}}t}\binits{B.}} \AND
\bauthor{\bsnm{Nechita},~\bfnm{Ion}\binits{I.}}
(\byear{2010}).
\btitle{Random quantum channels {I}: Graphical calculus and the Bell state
phenomenon}.
\bjournal{Comm. Math. Phys.}
\bvolume{297}
\bpages{345--370}.
\bid{doi={10.1007/s00220-010-1012-0}, mr={2651902}}
\bptnote{check year}
\end{barticle}
%
\endbibitem

%b10 ###
\bibitem{collins-nechita-2}
%
\begin{barticle}[auto:SpringerTagBib|2010-11-03|08:32:12]
\bauthor{\bsnm{Collins},~\bfnm{B.}\binits{B.}} \AND
\bauthor{\bsnm{Nechita},~\bfnm{I.}\binits{I.}}
(\byear{2011}).
\btitle{Random quantum channels II: Entanglement of random
subspaces, R\'enyi
entropy estimates and additivity problems.}
\bjournal{Adv. Math.}
\bvolume{226}
\bpages{1181--1201}.
\end{barticle}
%
\endbibitem

%b11 ###
\bibitem{collins-nechita-4}
%
\begin{barticle}[auto:SpringerTagBib|2010-11-03|08:32:12]
\bauthor{\bsnm{Collins},~\bfnm{B.}\binits{B.}} \AND
\bauthor{\bsnm{Nechita},~\bfnm{I.}\binits{I.}}
(\byear{2010}).
\btitle{Eigenvalue and entropy statistics for products of conjugate
random quantum channels}.
\bjournal{Entropy}
\bvolume{12}
\bpages{1612--1631}.
\end{barticle}
%
\endbibitem

%b12 ###
\bibitem{collins-nechita-zyczkowski}
%
\begin{barticle}[auto:SpringerTagBib|2010-11-03|08:32:12]
\bauthor{\bsnm{Collins},~\bfnm{B.}\binits{B.}},
\bauthor{\bsnm{Nechita},~\bfnm{I.}\binits{I.}} \AND
\bauthor{\bsnm{\.Zyczkowski},~\bfnm{K.}\binits{K.}}
(\byear{2010}).
\btitle{Random graph states, maximal flow and Fuss--Catalan
distributions}.
\bjournal{J. Phys. A~Math. Theor.}
\bvolume{43}
\bpages{275303}.
\end{barticle}
%
\endbibitem

%b13 ###
\bibitem{collins-sniady}
%
\begin{barticle}[mr]
\bauthor{\bsnm{Collins},~\bfnm{Beno{\^{\i}}t}\binits{B.}} \AND
\bauthor{\bsnm{{\'S}niady},~\bfnm{Piotr}\binits{P.}}
(\byear{2006}).
\btitle{Integration with respect to the {H}aar measure on unitary, orthogonal
and symplectic group}.
\bjournal{Comm. Math. Phys.}
\bvolume{264}
\bpages{773--795}.
\bid{doi={10.1007/s00220-006-1554-3}, mr={2217291}}
\end{barticle}
%
\endbibitem

%b14 ###
\bibitem{fukuda-king}
%
\begin{barticle}[mr]
\bauthor{\bsnm{Fukuda},~\bfnm{Motohisa}\binits{M.}} \AND
\bauthor{\bsnm{King},~\bfnm{Christopher}\binits{C.}}
(\byear{2010}).
\btitle{Entanglement of random subspaces via the {H}astings bound}.
\bjournal{J. Math. Phys.}
\bvolume{51}
\bpages{042201}.
\bid{doi={10.1063/1.3309418}, mr={2662469}}
\bptnote{check year}
\end{barticle}
%
\endbibitem

%b15 ###
\bibitem{fukuda-king-moser}
%
\begin{barticle}[mr]
\bauthor{\bsnm{Fukuda},~\bfnm{Motohisa}\binits{M.}},
\bauthor{\bsnm{King},~\bfnm{Christopher}\binits{C.}} \AND
\bauthor{\bsnm{Moser},~\bfnm{David~K.}\binits{D.~K.}}
(\byear{2010}).
\btitle{Comments on {H}astings' additivity counterexamples}.
\bjournal{Comm. Math. Phys.}
\bvolume{296}
\bpages{111--143}.
\bid{doi={10.1007/s00220-010-0996-9}, mr={2606630}}
\bptnote{check year}
\end{barticle}
%
\endbibitem

%b16 ###
\bibitem{graczyk}
%
\begin{barticle}[mr]
\bauthor{\bsnm{Graczyk},~\bfnm{Piotr}\binits{P.}},
\bauthor{\bsnm{Letac},~\bfnm{G{\'e}rard}\binits{G.}} \AND
\bauthor{\bsnm{Massam},~\bfnm{H{\'e}l{\`e}ne}\binits{H.}}
(\byear{2003}).
\btitle{The complex {W}ishart distribution and the symmetric group}.
\bjournal{Ann. Statist.}
\bvolume{31}
\bpages{287--309}.
\bid{doi={10.1214/aos/1046294466}, mr={1962508}}
\end{barticle}
%
\endbibitem

%b17 ###
\bibitem{guionnet}
%
\begin{bbook}[mr]
\bauthor{\bsnm{Guionnet},~\bfnm{Alice}\binits{A.}}
(\byear{2009}).
\btitle{Large Random Matrices: Lectures on Macroscopic Asymptotics}.
\bseries{Lecture Notes in Math.}
\bvolume{1957}.
\bpublisher{Springer}, \baddress{Berlin}.
%2006}.
\bid{doi={10.1007/978-3-540-69897-5}, mr={2498298}}
\end{bbook}
%
\endbibitem

%b18 ###
\bibitem{hanlon}
%
\begin{bincollection}[mr]
\bauthor{\bsnm{Hanlon},~\bfnm{Philip~J.}\binits{P.~J.}},
\bauthor{\bsnm{Stanley},~\bfnm{Richard~P.}\binits{R.~P.}} \AND
\bauthor{\bsnm{Stembridge},~\bfnm{John~R.}\binits{J.~R.}}
(\byear{1992}).
\btitle{Some combinatorial aspects of the spectra of normally distributed
random matrices}.
In \bbooktitle{Hypergeometric Functions on Domains of Positivity, {J}ack
Polynomials, and Applications ({T}ampa, {FL}, 1991)}.
\bseries{Contemp. Math.}
\bvolume{138}
\bpages{151--174}.
\bpublisher{Amer. Math. Soc.}, \baddress{Providence, RI}.
\bid{mr={1199126}}
\end{bincollection}
%
\endbibitem

%b19 ###
\bibitem{hastings}
%
\begin{barticle}[auto:SpringerTagBib|2010-11-03|08:32:12]
\bauthor{\bsnm{Hastings},~\bfnm{M.~B.}\binits{M.~B.}}
(\byear{2009}).
\btitle{Superadditivity of communication capacity using entangled
inputs}.
\bjournal{Nature Physics}
\bvolume{5}
\bpages{255}.
\end{barticle}
%
\endbibitem

%b20 ###
\bibitem{hayden-winter}
%
\begin{barticle}[mr]
\bauthor{\bsnm{Hayden},~\bfnm{Patrick}\binits{P.}} \AND
\bauthor{\bsnm{Winter},~\bfnm{Andreas}\binits{A.}}
(\byear{2008}).
\btitle{Counterexamples to the maximal {$p$}-norm multiplicity
conjecture for
all {$p>1$}}.
\bjournal{Comm. Math. Phys.}
\bvolume{284}
\bpages{263--280}.
\bid{doi={10.1007/s00220-008-0624-0}, mr={2443305}}
\end{barticle}
%
\endbibitem

%b21 ###
\bibitem{jones}
%
\begin{bmisc}[auto:SpringerTagBib|2010-11-03|08:32:12]
\bauthor{\bsnm{Jones},~\bfnm{V.~F.~R.}\binits{V.~F.~R.}}
(\byear{1999}).
\bhowpublished{\textit{Planar Algebras}. Available at}
\href{http://arxiv.org/abs/math/9909027v1}{arXiv:math/9909027v1}.
\end{bmisc}
%
\endbibitem

%b22 ###
\bibitem{mingo-nica}
%
\begin{barticle}[mr]
\bauthor{\bsnm{Mingo},~\bfnm{James~A.}\binits{J.~A.}} \AND
\bauthor{\bsnm{Nica},~\bfnm{Alexandru}\binits{A.}}
(\byear{2004}).
\btitle{Annular noncrossing permutations and partitions, and second-order
asymptotics for random matrices}.
\bjournal{Int. Math. Res. Not.}
\bvolume{28}
\bpages{1413--1460}.
\bid{doi={10.1155/S1073792804133023}, mr={2052516}}
\end{barticle}
%
\endbibitem

%b23 ###
\bibitem{nechita}
%
\begin{barticle}[mr]
\bauthor{\bsnm{Nechita},~\bfnm{Ion}\binits{I.}}
(\byear{2007}).
\btitle{Asymptotics of random density matrices}.
\bjournal{Ann. Henri Poincar\'e}
\bvolume{8}
\bpages{1521--1538}.
\bid{doi={10.1007/s00023-007-0345-5}, mr={2374950}}
\end{barticle}
%
\endbibitem

%b24 ###
\bibitem{nica-speicher}
%
\begin{bbook}[mr]
\bauthor{\bsnm{Nica},~\bfnm{Alexandru}\binits{A.}} \AND
\bauthor{\bsnm{Speicher},~\bfnm{Roland}\binits{R.}}
(\byear{2006}).
\btitle{Lectures on the Combinatorics of Free Probability}.
\bseries{London Mathematical Society Lecture Note Series}
\bvolume{335}.
\bpublisher{Cambridge Univ. Press}, \baddress{Cambridge}.
\bid{doi={10.1017/CBO9780511735127}, mr={2266879}}
\end{bbook}
%
\endbibitem

%b25 ###
\bibitem{page}
%
\begin{barticle}[mr]
\bauthor{\bsnm{Page},~\bfnm{Don~N.}\binits{D.~N.}}
(\byear{1993}).
\btitle{Average entropy of a subsystem}.
\bjournal{Phys. Rev. Lett.}
\bvolume{71}
\bpages{1291--1294}.
\bid{doi={10.1103/PhysRevLett.71.1291}, mr={1232812}}
\end{barticle}
%
\endbibitem

%b26 ###
\bibitem{sz1}
%
\begin{barticle}[mr]
\bauthor{\bsnm{Sommers},~\bfnm{Hans-J{\"u}rgen}\binits{H.-J.}} \AND
\bauthor{\bsnm{{\.Z}yczkowski},~\bfnm{Karol}\binits{K.}}
(\byear{2004}).
\btitle{Statistical properties of random density matrices}.
\bjournal{J. Phys. A}
\bvolume{37}
\bpages{8457--8466}.
\bid{doi={10.1088/0305-4470/37/35/004}, mr={2091254}}
\end{barticle}
%
\endbibitem

%b27 ###
\bibitem{stanley}
%
\begin{bbook}[mr]
\bauthor{\bsnm{Stanley},~\bfnm{Richard~P.}\binits{R.~P.}}
(\byear{1997}).
\btitle{Enumerative Combinatorics. {V}ol. 1}.
\bseries{Cambridge Studies in Advanced Mathematics}
\bvolume{49}.
\bpublisher{Cambridge Univ. Press}, \baddress{Cambridge}.
%original}.
\bid{mr={1442260}}
\end{bbook}
%
\endbibitem

%b28 ###
\bibitem{zycsommers}
%
\begin{barticle}[mr]
\bauthor{\bsnm{{\.Z}yczkowski},~\bfnm{Karol}\binits{K.}} \AND
\bauthor{\bsnm{Sommers},~\bfnm{Hans-J{\"u}rgen}\binits{H.-J.}}
(\byear{2001}).
\btitle{Induced measures in the space of mixed quantum states}.
\bjournal{J. Phys. A}
\bvolume{34}
\bpages{7111--7125}.
\bnote{Quantum information and computation}.
\bid{doi={10.1088/0305-4470/34/35/335}, mr={1863143}}%
\end{barticle}\
%
\endbibitem

%b29 ###
\bibitem{zvonkin}
%
\begin{barticle}[mr]
\bauthor{\bsnm{Zvonkin},~\bfnm{A.}\binits{A.}}
(\byear{1997}).
\btitle{Matrix integrals and map enumeration: An accessible introduction}.
\bjournal{Math. Comput. Modelling}
\bvolume{26}
\bpages{281--304}.
\bid{doi={10.1016/S0895-7177(97)00210-0}, mr={1492512}}
\end{barticle}
%
\endbibitem

\end{thebibliography}
\end{document}